\documentclass[10pt,a4paper]{amsart}
\usepackage[margin=20mm]{geometry}
\usepackage[numbers]{natbib}
\usepackage{hyperref}

\usepackage{amssymb,amsmath}
\usepackage{amsthm}
\usepackage{bbm}
\usepackage{stmaryrd}
\usepackage[usenames,dvipsnames]{xcolor}
\usepackage{color}

\input xy
\xyoption{all} 

\numberwithin{equation}{section}

\newtheorem{theorem}{Theorem}[section]

\newtheorem{lemma}[theorem]{Lemma}
\newtheorem{proposition}[theorem]{Proposition}

\theoremstyle{definition}
\newtheorem{definition}[theorem]{Definition}
\newtheorem{remark}[theorem]{Remark}
\newtheorem{example}[theorem]{Example}

\newenvironment{warning}[1][Warning.]{\begin{trivlist}
\item[\hskip \labelsep {\bfseries #1}]}{\end{trivlist}}

\newcommand{\InHom}{\mbox{$\underline{\Hom}$}}

\newcommand{\rmd}{\textnormal{d}}

\DeclareMathOperator{\Vect}{Vect}

\DeclareMathOperator{\Der}{Der}

\DeclareMathOperator{\Hom}{Hom}

\DeclareMathOperator{\Ker}{Ker}

\newcommand{\catname}[1]{\textnormal{\texttt{#1}}}

\font\black=cmbx10 \font\sblack=cmbx7 \font\ssblack=cmbx5 \font\blackital=cmmib10  \skewchar\blackital='177
\font\sblackital=cmmib7 \skewchar\sblackital='177 \font\ssblackital=cmmib5 \skewchar\ssblackital='177
\font\sanss=cmss10 \font\ssanss=cmss8 
\font\sssanss=cmss8 scaled 600 \font\blackboard=msbm10 \font\sblackboard=msbm7 \font\ssblackboard=msbm5
\font\caligr=eusm10 \font\scaligr=eusm7 \font\sscaligr=eusm5  \font\fraktur=eufm10
\font\sfraktur=eufm7 \font\ssfraktur=eufm5 
\font\bsymb=cmsy10 scaled\magstep2
\def\all#1{\setbox0=\hbox{\lower1.5pt\hbox{\bsymb
       \char"38}}\setbox1=\hbox{$_{#1}$} \box0\lower2pt\box1\;}
\def\exi#1{\setbox0=\hbox{\lower1.5pt\hbox{\bsymb \char"39}}
       \setbox1=\hbox{$_{#1}$} \box0\lower2pt\box1\;}

\def\tx#1{{\fam0\relax#1}}

\newfam\bifam
\textfont\bifam=\blackital \scriptfont\bifam=\sblackital \scriptscriptfont\bifam=\ssblackital

\newfam\blfam
\textfont\blfam=\black \scriptfont\blfam=\sblack \scriptscriptfont\blfam=\ssblack

\newfam\bbfam
\textfont\bbfam=\blackboard \scriptfont\bbfam=\sblackboard \scriptscriptfont\bbfam=\ssblackboard

\newfam\ssfam
\textfont\ssfam=\sanss \scriptfont\ssfam=\ssanss \scriptscriptfont\ssfam=\sssanss
\def\sss#1{{\fam\ssfam\relax#1}}

\newfam\clfam
\textfont\clfam=\caligr \scriptfont\clfam=\scaligr \scriptscriptfont\clfam=\sscaligr

\newfam\frfam
\textfont\frfam=\fraktur \scriptfont\frfam=\sfraktur \scriptscriptfont\frfam=\ssfraktur

\def\hpb#1{\setbox0=\hbox{${#1}$}
    \copy0 \kern-\wd0 \kern.2pt \box0}
\def\vpb#1{\setbox0=\hbox{${#1}$}
    \copy0 \kern-\wd0 \raise.08pt \box0}

\def\pmb#1{\setbox0\hbox{${#1}$} \copy0 \kern-\wd0 \kern.2pt \box0}
\def\pmbb#1{\setbox0\hbox{${#1}$} \copy0 \kern-\wd0
      \kern.2pt \copy0 \kern-\wd0 \kern.2pt \box0}
\def\pmbbb#1{\setbox0\hbox{${#1}$} \copy0 \kern-\wd0
      \kern.2pt \copy0 \kern-\wd0 \kern.2pt
    \copy0 \kern-\wd0 \kern.2pt \box0}
\def\pmxb#1{\setbox0\hbox{${#1}$} \copy0 \kern-\wd0
      \kern.2pt \copy0 \kern-\wd0 \kern.2pt
      \copy0 \kern-\wd0 \kern.2pt \copy0 \kern-\wd0 \kern.2pt \box0}
\def\pmxbb#1{\setbox0\hbox{${#1}$} \copy0 \kern-\wd0 \kern.2pt
      \copy0 \kern-\wd0 \kern.2pt
      \copy0 \kern-\wd0 \kern.2pt \copy0 \kern-\wd0 \kern.2pt
      \copy0 \kern-\wd0 \kern.2pt \box0}

\mathchardef\za="710B  
\mathchardef\zb="710C  
\mathchardef\zg="710D  
\mathchardef\zd="710E  
\mathchardef\zve="710F 
\mathchardef\zz="7110  
\mathchardef\zh="7111  
\mathchardef\zvy="7112 
\mathchardef\zi="7113  
\mathchardef\zk="7114  
\mathchardef\zl="7115  
\mathchardef\zm="7116  
\mathchardef\zn="7117  
\mathchardef\zx="7118  
\mathchardef\zp="7119  
\mathchardef\zr="711A  
\mathchardef\zs="711B  
\mathchardef\zt="711C  
\mathchardef\zu="711D  
\mathchardef\zvf="711E 
\mathchardef\zq="711F  
\mathchardef\zc="7120  
\mathchardef\zw="7121  
\mathchardef\ze="7122  
\mathchardef\zy="7123  
\mathchardef\zf="7124  
\mathchardef\zvr="7125 
\mathchardef\zvs="7126 
\mathchardef\zf="7127  
\mathchardef\zG="7000  
\mathchardef\zD="7001  
\mathchardef\zY="7002  
\mathchardef\zL="7003  
\mathchardef\zX="7004  
\mathchardef\zP="7005  
\mathchardef\zS="7006  
\mathchardef\zU="7007  
\mathchardef\zF="7008  
\mathchardef\zW="700A  
\mathchardef\zC="7009  

\newcommand{\be}{\begin{equation}}
\newcommand{\ee}{\end{equation}}

\newcommand{\bea}{\begin{eqnarray}}
\newcommand{\eea}{\end{eqnarray}}
\def\*{{\textstyle *}}
\newcommand{\R}{{\mathbb R}}

\newcommand{\Z}{{\mathbb Z}}

\newcommand{\s}{{\textstyle *}}

\newcommand{\pa}{\partial}

\def\ran{\rangle}

\def\Hom{\sss{Hom}}

\def\Vect{\sss{Vect}}

\def\la{\langle}
\def\ran{\rangle}

\def\sT{{\sss T}}

\def\xi{\tx{i}}

\def\xd{\operatorname{d}}

\def\s*{{\scriptstyle *}}

\def\cO{\mathcal{O}}

\newcommand{\beas}{\begin{eqnarray*}}
\newcommand{\eeas}{\end{eqnarray*}}

\def\half{\frac{1}{2}}

\title{Symplectic $\Z_2^n$-manifolds }

   \author{Andrew James Bruce$^\dag$ \& Janusz Grabowski$^\ddag$}
 \address{\hskip-.4cm  $^\dag$ Mathematics Research Unit, University of Luxembourg,  Esch-sur-Alzette, Luxembourg \\
   \newline  $^\ddag$ Institute of Mathematics, Polish Academy of Sciences, Poland}
   \email{andrewjamesbruce@googlemail.com,~jagrab@impan.pl}

\date{\today}
 
\begin{document}

\begin{abstract}
Roughly speaking, $\Z_2^n$-manifolds are `manifolds' equipped with $\Z_2^n$-graded commutative coordinates with the sign rule being determined by the scalar product of their $\Z_2^n$-degrees. We examine the notion of a symplectic $\Z_2^n$-manifold, i.e., a $\Z_2^n$-manifold equipped with a symplectic two-form that may carry non-zero $\Z_2^n$-degree. We show that the basic notions and results of symplectic geometry generalise to the `higher graded' setting, including a generalisation of Darboux's theorem. \par
\end{abstract}

 \maketitle

\setcounter{tocdepth}{2}
 \tableofcontents

\section{Introduction}
It is well known that symplectic geometry, and more generally Poisson geometry, is part of the natural geometric framework for classical mechanics. The impact of this on both mathematics and physics hardly needs to be emphasised here. There are many different generalisations of symplectic geometry to be found in the literature, but important for the work are symplectic supermanifolds. One non-classical aspect of supergeometry is Grassmann odd geometric structures, such objects have no true classical counterpart. In  particular, the category of supermanifolds allows for even and odd symplectic structures, the latter being vital in the BV-BRST formalism (see \cite{Khudaverdian:2004,Schwarz:1993,Schwarz:1996}). The notion of an odd connection on a supermanifold was given by the authors in a previous paper \cite{Bruce:2020bb}. Even and odd Riemannian structures in supermanifolds have long been studied (see for example \cite{Asorey:2009,Galaev:2012,Garnier:2012,Goertsches:2008}). \par
In this paper, we examine symplectic structures in the rapidly developing setting of $\Z_2^n$-geometry ($\Z_2^n := \Z_2 \times \cdots \times \Z_2 $), see \cite{Bruce:2020,Bruce:2018,Bruce:2018b,Covolo:2012,Covolo:2016a,Covolo:2016b,Poncin:2016}. Naturally, a symplectic $\Z_2^n$-manifold is a  $\Z_2^n$-manifold equipped with a non-degenerate closed two form, we will make all this precise in due course.  As compared with supergeometry, there is a lot more freedom with the degree of a symplectic structure beyond simply even or odd, here understood as the total degree. We remark that Riemannian $\Z_2^n$-manifolds were the subject of our previous paper \cite{Bruce:2020b}, and in complete parallel with the classical setting, can be viewed as the symmetric cousins of symplectic $\Z_2^n$-manifolds.\par
$\Z_2^n$-manifolds are a very natural, but demanding generalisation of supermanifolds, both of which are `species' of locally ringed spaces.  Heuristically,  $\Z_2^n$-manifolds are `manifolds' with local coordinates whose sign rule is determined by their $\Z_2^n$-degree and the standard scalar product in $\Z_2^n$:
$$\la(a_1\dots,a_n),(b_1,\dots, b_n)\ran=\sum_ia_ib_i\,.$$
This implies that for $n>1$, in general, we have formal coordinates that are not nilpotent. The non-nilpotent nature of some of the formal coordinates implies that we must work with formal power series and not simply polynomials. In practice, this means that although many of the standard results from the theory of supermanifolds hold for $\Z_2^n$-manifolds verbatim, the proofs are often significantly more involved. It has also been shown that a  $\Z_2^n$-manifold is fully encoded in its algebra of global sections (see \cite[Corollary 3.8]{Bruce:2018b}), and this algebra is an example of an almost commutative algebra in the sense of  Bongaarts \& Pijls \cite{Bongaarts:1994}. We also have a chart theorem (\cite[Theorem 7.10]{Covolo:2016}) that says that we can work with morphisms between $\Z_2^n$-manifolds at the level of local coordinates.  Thus, one should view  $\Z_2^n$-manifolds as a very workable form of noncommutative algebraic geometry. We also remark  that $\Z_2^n$-geometry sits comfortably within Majid's framework of braided geometry (for an introduction see \cite{Majid:1995}).\par
 In this paper we establish, amongst other results, the following:
\begin{itemize}
\item If the symplectic structure on a $\Z_2^n$-manifold is of $\Z_2^n$-degree zero, i.e., of $\Z_2^n$-degree $\mathbf{0} = (0,0, \cdots, 0)$, then the underlying reduced manifold is itself a symplectic manifold (see Proposition \ref{prop:RedStruc}). There is no analogue statement for structures of non-zero $\Z_2^n$-degree.
\item The algebra of global sections on a $\Z_2^n$-manifold comes equipped with  a $\Z_2^n$-graded Poisson bracket (see Proposition \ref{prop:AlmPoissbrk} and Theorem \ref{thm:PoissbrkJac}).
\item Appropriately degree shifted cotangent bundles of $\Z_2^n$-manifolds come with canonical symplectic structures (see Theorem \ref{thm:CoTanStruct}).
\item Any symplectic $\Z_2^n$-manifold admits $\Z_2^n$-graded Darboux coordinates. That is, locally the symplectic structure can be brought into canonical form (see Theorem \ref{thm:Darboux}).
\end{itemize}
Parallel  to the mathematical works, there has been a renewed interest in the physical applications of $\Z_2^n$-gradings, for example \cite{Aizawa:2020a,Aizawa:2020b,Aizawa:2016,Aizawa:2020c,Aizawa:2020d,Bruce:2019a,Bruce:2020aa,Bruce:2020a,Bruce:2019aa,Tolstoy:2014b}. The origin of $\Z_2^n$-graded associative algebras can be traced back to  Rittenberg \& Wyler \cite{Rittenberg:1978} and Scheunert \cite{Scheunert:1979}, both works were motivated by $\Z_2^n$-Lie algebras and their potential applications in physics.  It has long been appreciated that $\Z_2 \times \Z_2$-gradings appear in parastatistics, see for example \cite{Tolstoy:2014,Yang:2001}.  Thus, we will focus on the mathematical aspects of the theory of symplectic  $\Z_2^n$-manifolds and will not explore the potential applications in physics other than making some comments in the discussion section of this paper.
\section{Symplectic $\Z_2^n$-manifolds}\label{sec:RiemZman}
\subsection{Essential elements of $\Z_2^n$-geometry}
Covolo, Grabowski and Poncin pioneered the locally ringed space approach to $\Z_2^n$-manifolds (see \cite{Covolo:2016a,Covolo:2016b}).  Further various technical issues have been clarified in \cite{Bruce:2018,Bruce:2018b,Covolo:2016}.  We will draw upon these works and not present proofs of any formal statements. To set notation, by $\Z_2^n$ we mean the abelian group $\Z_2 \times \Z_2 \times \cdots \times \Z_2$ where the Cartesian product is over $n$ factors.  Once an ordering has been fixed, we will  denote elements of $\Z_2^n$ as $\gamma_i$ for $i = 0 ,1 \cdots , N$, where  $N = 2^n-1$. The convention of ordering that we will choose is to fill in the zeros and the ones from the left,  and then  placing the even elements first and then the odd elements. For example, with this choice of ordering
\begin{align*}
& \Z_{2}^{2} = \{ (0,0),  \: (1,1), \: (0,1), \: (1,0)\},\\
& \Z_2^3 = \{ (0, 0, 0), \:  (0, 1, 1), \: (1, 0, 1), \: (1, 1, 0),\:  (0, 0, 1), \: (0, 1, 0), \: (1, 0, 0),  \:  (1, 1, 1) \}.
\end{align*}
In particular, $\gamma_0 = (0,0, \cdots, 0) =: \mathbf{0}$. We set $\mathbf{q} := (q_{1}, q_{2}, \cdots , q_{N})$, where  $N = 2^n-1$ and $q_i  \in \mathbb{N}$.
\begin{definition}
A \emph{locally} $\Z_{2}^{n}$-\emph{ringed space}, $n \in \mathbb{N} \setminus \{0\}$, is a pair $S := (|S|, \cO_S )$ where $|S|$ is a second-countable Hausdorff space, and $\cO_{S}$  is a sheaf  of $\Z_{2}^{n}$-graded $\Z_{2}^{n}$-commutative associative unital $\R$-algebras, such that the stalks $\cO_{S,p}$, $p \in  |S|$ are local rings.
\end{definition}
In the above context, $\Z_2^n$-commutative means that any two  (local) sections $s$, $t \in \cO_{S}(|U|)$, $|U| \subset |S|$ open, of homogeneous degree $\deg(s)  \in \Z_{2}^{n}$ and $\deg(t) \in \Z_2^n$ commute up to the sign rule
$$s\,t = (-1)^{\langle \deg(s), \deg(t)\rangle} \: t\,s,$$
where $\langle \: , \:\rangle$ is the standard scalar product on $\Z_2^n$. For example, consider the  $\gamma_3 =(0,1)$ and $\gamma_2 =(1,1) \in \Z_2^2$, then $\langle \gamma_3, \gamma_2 \rangle  = 0 \times 1 + 1 \times 1 =1$.
\begin{definition}[\cite{Bruce:2020}]\label{def:Z2nGrassmann}
 A  $\Z_2^n$-\emph{Grassmann algebra}  $\Lambda^{\mathbf{q}} := \R[[\zx]]$ consists of a formal power series generated by the $\Z_2^n$-graded variables $\{ \zx^A \}$, there are $q_i$ generators of degree $\gamma_i \in \Z_2^n$,  ($i >0$), subject to the  relation
$$\zx^A \zx^B = (-1)^{\langle \deg(A) , \deg(B)  \rangle } \zx^B \zx^A,$$
where $\deg(\zx^A) =: \deg(A) \in \Z_2^n\setminus \{ 0\}$ and similar for $\zx^B$.
\end{definition}
A supermanifold  can heuristically be thought of as a manifold for which  the structure sheaf is enriched by a Grassmann algebra. Similarly,  a $\Z_2^n$-manifold can, very loosely,  be thought of as  a manifold whose structure sheaf is modified to include generators of a $\Z_2^n$-Grassmann algebra.
\begin{definition}[\cite{Covolo:2016a}]
A (smooth) $\Z_2^n$-\emph{manifold} of dimension $p |\mathbf{q}$ is a locally $\Z_{2}^{n}$-ringed space $ M := \left(|M|, \cO_M \right)$, which is locally isomorphic to the $\Z_2^n$-ringed space $\mathbb{R}^{p |\mathbf{q}} := \left( \mathbb{R}^{p}, C^{\infty}_{\mathbb{R}^{p}}(-)[[\zx]] \right)$. Local sections of $M$ are formal power series in the $\Z_2^n$-graded variables $\zx$ with  smooth coefficients,
$$\cO_M(|U|) \simeq C^{\infty}(|U|)[[\zx]] :=  \left \{ \sum_{\alpha \in \mathbb{N}^{N}}  \zx^{\alpha}f_{\alpha} ~ | \: f_{ \alpha} \in C^{\infty}(|U|)\right \},$$
for `small enough' opens $|U|\subset |M|$.   \emph{Morphisms} between $\Z_{2}^{n}$-manifolds are  morphisms of $\Z_{2}^{n}$-ringed spaces, that is,  pairs $(\phi, \phi^{*}) : (|M|, \cO_M) \rightarrow  (|N|, \cO_N)$ consisting of a continuous map  $\phi: |M| \rightarrow |N|$ and sheaf morphism $\phi^{*}: \cO_N(|V|) \rightarrow \cO_M(\phi^{-1}(|V|))$, where $|V| \subset |N|$  is open. The global sections of the structure sheaf $\cO_M$ will be referred to as \emph{functions} and we denote the algebra of global sections as $C^{\infty}(M) := \cO_{M}(|M|)$.
\end{definition}
\begin{example}[The local model]\label{exp:SuperDom}
The locally $\Z_{2}^{n}$-ringed space $\mathcal{U}^{p|\mathbf{q}} :=  \big(\mathcal{U}^p , C^\infty_{\mathcal{U}^p}(-)[[\zx]] \big)$, where $\mathcal{U}^p \subseteq \R^p$ (open) is naturally a $\Z_2^n$-manifold. Such  $\Z_2^n$-manifolds are referred to  as \emph{$\Z_2^n$-domains} of dimension $p|\mathbf{q}$.  We can employ (natural) coordinates $x^I := (x^a, \zx^A)$ on any $\Z_2^n$-domain, where $x^a$ form a coordinate system on $\mathcal{U}^p$ and the $\zx^A$ are formal coordinates.
\end{example}
Associated with any $\Z_{2}^{n}$-graded algebra $\mathcal{A}$ is the  ideal $J$ of $\mathcal{A}$, generated by the homogeneous elements of $\mathcal{A}$ that have non-zero $\Z_{2}^{n}$-degree. The associated $J$-adic topology is vital in the theory of $\Z_{2}^{n}$-manifolds. In particular, given a morphism of $\Z_{2}^{n}$-graded algebras $\psi : \mathcal{A} \rightarrow \mathcal{A}^{\prime}$, then $\psi (J_{\mathcal{A}} ) \subset J_{\mathcal{A}^{\prime}}$.  These notions can be `sheafified', i.e., for any $\Z_{2}^{n}$-manifold $M$, there exists an \emph{ideal sheaf} $\mathcal{J}_M$, defined by
$\mathcal{J}(|U| ) := \langle f \in \cO_M(|U|)~|~ \deg(f)\neq 0 \rangle$. The $\mathcal{J}_M$-adic topology on $\cO_M$ can then be defined in the obvious way. It is known that the structure sheaf of a $\Z_2^n$-manifold $\cO_M$ is Hausdorff complete with respect to the $\mathcal{J}_M$-adic topology (see \cite[Proposition 7.9]{Covolo:2016} for details). Furthermore, for any $\Z_2^n$-manifold~$M$, there exists a short exact sequence of sheaves of $\Z_2^n$-graded $\Z_2^n$-commutative associative $\R$-algebras
\begin{equation}\label{eqn:SES}
0\longrightarrow\ker\zve\longrightarrow\cO_M\stackrel{\zve}{\longrightarrow}C^\infty_{|M|}\longrightarrow 0,
\end{equation}
such that $\ker \zve=\mathcal{J}_M$. Informally, $\zve_{|U|}: \cO_M(|U|) \rightarrow C^\infty(|U|)$ is simply ``forgetting" the formal coordinates. The \emph{reduced manifold} is defined as $M_{red} := \big(|M|, \, C^\infty_{|M|}\big)$. As is customary in supergeometry, we will usually simply write $|M|$ for the reduced manifold. Note that the reduced manifold is a $\Z_2^n$-submanifold  $ |M|\hookrightarrow M $ defined by the dual of $\zve$. This inclusion is, of course, over the identity on $|M|$. Locally this inclusion is given by $\zve( x^a) = x^a$ and $\zve (\zx^A)=0$.\par
As standard on manifolds and supermanifolds, one can build global geometric concepts via the gluing of local geometric objects. This is a particularly useful point of view when constructing new $\Z_2^n$-manifolds from old ones, for example, Cartesian products and fibre bundles.  That is, we can consider a $\Z_2^n$-manifold as being covered by  $\Z_2^n$-domains together with the appropriate coordinate transformations.   The key result here is the \emph{chart theorem} (\cite[Theorem 7.10]{Covolo:2016}) that allows us to write morphisms of $\Z_2^n$-manifolds in terms of the local coordinates.  Specifically, suppose we have two  $\Z_2^n$-domains $\mathcal{U}^{p|\mathbf{q}}$ and $\mathcal{V}^{r|\mathbf{s}}$. Morphisms $\phi: \mathcal{U}^{p|\mathbf{q}} \longrightarrow \mathcal{V}^{r|\mathbf{s}}$  correspond  to \emph{graded unital $\R$-algebra morphisms}
 \begin{equation*}
 \phi^* : C^{\infty}\big(\mathcal{V}^r \big)[[\eta]] \longrightarrow   C^{\infty}\big(\mathcal{U}^p \big)[[\zx]],
 \end{equation*}
which are fully determined by their coordinate expressions. We now proceed to describe atlases as these will be useful when describing differential forms via a generalised notion of a tangent bundle, for example.\par
For any $M := \big( |M|, \cO_M \big)$, we  define open $\Z_2^n$-submanifolds  of $M$ as  $\Z_2^n$-manifolds of the form $U := \big(|U|, \cO_M|_{|U|} \big)$, where $|U|\subset |M|$ is open. Due to the local structure of a  $\Z_2^n$-manifold  of dimension $p|\mathbf{q}$ we know that for a `small enough' $|U| \subset |M|$ there exists an isomorphism of $\Z_2^n$-manifolds
$$\psi : U \longrightarrow  \mathcal{U}^{p|\mathbf{q}}.$$
We refer to a pair $(U, \psi)$  as a \emph{(coordinate) chart}, and a family of charts $\{(U_i , \psi_i)  \}_{i \in \mathcal{I}}$ we refer to as an \emph{atlas}  if the family $\{ |U_i| \}_{i \in \mathcal{I}}$ forms an open cover of $|M|$.  Coordinate transformations are constructed as follows. Suppose we have two charts
\begin{align*}
&\psi_i :  U_i \stackrel{\sim}{\rightarrow} \mathcal{U}_i^{p|\mathbf{q}}= \big(\mathcal{U}^p_i, C^\infty_{\mathcal{U}^p_i}[[\zx]] \big), && \psi_j :  U_j \stackrel{\sim}{\rightarrow} \mathcal{U}_j^{p|\mathbf{q}}= \big(\mathcal{U}^p_j, C^\infty_{\mathcal{U}^p_j}[[\eta]]\big),
\end{align*}
where we assume that $|U_{ij}|:= |U_i|\cap |U_j| \neq \emptyset $. We then set
\begin{align*}
& |\psi_i|\big( |U_{ij}|\big):= \mathcal{U}^p_{ij}, &&|\psi_j|\big( |U_{ij}|\big):= \mathcal{U}^p_{ji}.
\end{align*}
As standard, we have a homomorphism (not writing out the obvious restrictions)
$$|\psi_{ij}| := |\psi_j|\circ |\psi_i|^{-1} : \mathcal{U}^p_{ij} \longrightarrow \mathcal{U}^p_{ji}.$$
We need to define what happens to sections, and via the Chart Theorem, we know it is sufficient to describe this in terms of the coordinates.  Specifically,
$$\psi^*_{ij} :=  (\psi_i^*)^{-1} \circ \psi_j^* :  C^\infty(\mathcal{U}^p_{ji})[[\eta]] \longrightarrow C^\infty(\mathcal{U}^p_{ji})[[\zx]],$$
provides the required $\Z_2^n$-graded  $\Z_2^n$-commutative algebra isomorphism.   In coordinates, we write $\psi^*_{ij}(y^b ,\eta^B) =  (y^b(x, \zx),  \eta^B(x, \zx))$ and similar.\par
 As a result of the local structure of a $\Z_{2}^{n}$-manifold any morphism $\phi : M \rightarrow N$  can be uniquely specified by a family of local morphisms between $\Z_{2}^{n}$-domains, once atlases on $M$ and $N$ have been fixed. Thus, morphisms of $\Z_{2}^{n}$-manifolds can be fully described using local coordinates. We will regularly exploit this and  employ the standard abuses of notation as found in classical differential geometry when writing morphisms using local coordinates.
\begin{remark}
There is  a Batchelor--Gaw\c{e}dzki theorem for  $\Z_{2}^{n}$-manifolds,  see \cite[Theorem 3.2]{Covolo:2016a}. The theorem states that any $\Z_2^n$-manifold is noncanonically isomorphic to a $\Z_2^n\setminus \{ \mathbf{0}\}$-graded vector bundle over  a smooth manifold, which is, of course, the reduced manifold of the $\Z_2^n$-manifold.  We will not use this theorem in this paper.
\end{remark}
The \emph{tangent sheaf} $\mathcal{T}M$ of a $\Z_2^n$-manifold $M$ is  defined as the sheaf of derivations of sections of the structure sheaf, i.e., $\mathcal{T}M(|U|) := \textnormal{Der}(\cO_M(|U|))$, for arbitrary $|U| \subset |M|$. Naturally, this is  a sheaf of locally free $\cO_{M}$-modules. Global sections of the tangent sheaf are referred  to as \emph{vector fields}. We denote the $\cO_{M}(|M|)$-module of vector fields as $\Vect(M)$. We can always `localise' vector fields, $X|_{|U|} =  X_{|U|} \in\textnormal{Der}(\cO_M(|U|)) $ (see \cite[Lemma 2.2]{Covolo:2016}). For `small enough' opens we can employ local coordinates $x^I = (x^a, \zx^A)$ and write
$$X_{|U|} =  X^I(x) \frac{\partial}{\partial x^I}.$$
We will generally not explicitly write restrictions of a vector field and so on when presenting them in local coordinates. Under changes of coordinates $x^{I'} = x^{I'}(x)$ the local components of a vector field transform as
$$X^{I'} =  X^I \left(\frac{\partial x^{I'}}{\partial x^I}  \right).$$
The $\R$-vector space $\Vect(M)$ forms a $\Z_2^n$-Lie algebra (see \cite[Section 2]{Covolo:2012}) with the Lie bracket being defined as
$$[X,Y] :=  X \circ Y -  (-1)^{\langle \deg(X), \deg(Y) \rangle} Y \circ X,$$
for homogeneous $X$ and $Y \in \Vect(M)$, extension to inhomogeneous vector fields is via linearity.  The reader can quickly verify that the Lie bracket is (graded) skewsymmetric:
$$[X,Y] =  - (-1)^{\langle \deg(X), \deg(Y) \rangle}\, [Y,X],$$
and satisfies the (graded) Jacobi identity
$$[X,[Y,Z]] =  [[X,Y],Z] +  (-1)^{\langle \deg(X), \deg(Y) \rangle}\,  [Y, [X,Z]],$$
where $X,Y$ and $Z \in \Vect(M)$ are homogeneous but otherwise arbitrary.

\subsection{The $\Z_2^{n+1}$-tangent functor and differential forms}
Recall that differential forms on a manifold $M_0$ can be defined as functions on the supermanifold $\Pi \sT M_0$, known as the antitangent or shifted tangent bundle. Differential forms on a supermanifold, in Bernstein's conventions, can similarly be defined using the antitangent bundle. In \cite{Covolo:2016}, it was suggested that differential forms on a supermanifold $M_1$ in Deligne's conventions can be understood as functions on a $\Z_2^2$-manifold, which we will denote as $\mathbb{T}M_1$. We will take the position that differential forms on a $\Z_2^n$-manifold can be defined as functions on a $\Z_2^{n+1}$-graded version of the tangent bundle. Differential forms on a $\Z_2^n$-manifold were first considered in \cite{Covolo:2016} and match our conventions. The new aspect is the understanding as functions on $\Z_2^{n+1}$-manifold. \par
We construct the \emph{$\Z_2^{n+1}$-tangent functor}
\begin{align*}
\mathbb{T} :~&  \Z_2^n\catname{Man} \longrightarrow \Z_2^{n+1}\catname{Man}\\
& M \mapsto \mathbb{T}M
\end{align*}
via the construction of an atlas coming from an atlas on $M$. We will then come to construct morphisms similarly.  To do this we first  consider $\mathbb{T}\mathcal{U}^{p|\mathbf{q}}_i$. The reduced manifold is defined to be $|\mathbb{T}\mathcal{U}^{p|\mathbf{q}}_i| = \mathcal{U}^p_i$.  We define
$$\mathbb{T}\mathcal{U}^{p|\mathbf{q}}_i := \big ( \mathcal{U}^p_i, \, C^{\infty}_{\mathcal{U}^p_i}[[\zx, \rmd x, \rmd \zx]]   \big ),$$
where the formal coordinates are assigned the $\Z_2^{n+1}$-degrees
$$\{\underbrace{\zx^A}_{(0, \deg(A))}, ~ \underbrace{\rmd x^b}_{(1, \mathbf{0})}, ~ \underbrace{\rmd \zx^B}_{(1, \deg(B))}  \}.$$
Note that $\zx$ are the formal coordinates on $\mathcal{U}^{p|\mathbf{q}}_i$, but now considered to be $\Z_2^{n+1}$-graded. Given an atlas $\{(U_i, \psi_i) \}_{i \in \mathcal{I}}$ on $M$, we need to construct $\{ (\mathbb{T}U_i , \mathbb{T}\psi_i) \}_{i \in \mathcal{I}}$. We assume $|U_{ij}| := |U_i|\cap |U_j| \neq \emptyset $, and set (as before)
\begin{align*}
& |\psi_i|\big( |U_{ij}|\big):= \mathcal{U}^p_{ij}, &&|\psi_j|\big( |U_{ij}|\big):= \mathcal{U}^p_{ji}.
\end{align*}
We then define
$$|\mathbb{T}\psi_{ij}| := |\psi_{ij}| :\mathcal{U}^p_{ij} \rightarrow \mathcal{U}^p_{ji}. $$
We now need to define the $\Z_2^{n+1}$-graded  $\Z_2^{n+1}$-commutative algebra morphism
$$\mathbb{T} \psi^*_{ij} : C^\infty(\mathcal{U}^p)_{ji}[[\eta, \rmd y, \rmd \eta]] \rightarrow  C^\infty(\mathcal{U}^p)_{ij}[[\zx, \rmd x, \rmd \zx]], $$
which we define using  $\mathbb{T}$ as a ``derivative'' and local coordinates:
\begin{align*}
& \mathbb{T}\psi^*_{ij}y^b := \psi^*_{ij}y^b = y^b(x, \zx), && \mathbb{T}\psi^*_{ij}\eta^B := \psi^*_{ij}\eta^B = \eta^B(x, \zx),\\
&\mathbb{T}\psi^*_{ij} \rmd y^c :=  \rmd\big(\psi^*_{ij}y^c\big) = \rmd x^a \frac{\partial y^c}{\partial x^a} + \rmd \zx^A \frac{\partial y^c}{\partial \zx^A},&& \mathbb{T}\psi^*_{ij} \rmd \eta^D :=  \rmd\big(\psi^*_{ij}\eta^D\big) = \rmd x^a \frac{\partial \eta^D}{\partial x^a} + \rmd \zx^A \frac{\partial \eta^D}{\partial \zx^A}.
\end{align*}
The cocycle condition can easily be checked and follows from the chain rule for derivatives (see \cite[Proposition 2.10]{Covolo:2016}).  The upshot is that $\mathbb{T}M$ comes with natural coordinates
$$(\underbrace{x^I}_{(0, \deg(I))}, ~ \underbrace{\rmd x^J}_{(1, \deg(J))}),$$
and the admissible coordinate transformations
\begin{align*}
& x^{I'} = x^{I'}(x), && \rmd x^{J'} =  \rmd x^{J}\left( \frac{\partial x^{J'}}{\partial x^J}\right).
\end{align*}
The commutation rules for the coordinates are
\begin{align*}
& x^I x^J =  (-1)^{\langle \deg(I), \deg(J) \rangle} \, x^J \ x^I, && x^I \rmd x^J =  (-1)^{\langle \deg(I), \deg(J) \rangle} \, \rmd x^J x^I, && \rmd x^I \rmd x^J =   - (-1)^{\langle \deg(I), \deg(J) \rangle} \, \rmd x^J  \rmd x^I.
\end{align*}
\begin{remark}
Clearly, $\mathbb{T}M$ is a vector bundle in the category of $\Z_2^{n+1}$-manifolds. We will not make use of this vector bundle structure and so postpone details of vector bundles in this higher graded setting to future publications.  We will only remark that categorical products exist in the category of $\Z_2^n$-manifolds (see \cite{Bruce:2018b}).
\end{remark}
The global algebra $\cO_{\mathbb{T}M}(|M|)$ is $\Z_2^{n+1}$-graded, $\Z_2^{n+1}$-commutative. Note that this is closest to Deligne's sign conventions for differential forms on a supermanifold.  Due to the linear nature of the coordinate transformations, the global algebra can be considered as $\mathbb{N}$-graded using the form degree or homological degree, i.e., locally the order in $\rmd x$ of any monomial. We then define differential forms on a $\Z_2^n$-manifold $M$ as global functions on $\mathbb{T}M$,
$$\Omega^*(M) := \cO_{\mathbb{T}M}(|M|) =  \bigoplus_{p \in \mathbb{N}} \big( \cO_{\mathbb{T}M}(|M|)\big)_p,$$
and $\Omega^p(M) :=  \big( \cO_{\mathbb{T}M}(|M|)\big)_p$. Note that $\Omega^0(M) \cong \cO_M(|M|)$.  It is then clear that the algebra of differential forms can naturally be thought of as a (right) $\cO_M(|M|)$-module.  Locally, any $p$-form looks like
$$\Omega^p(M) \ni \alpha =  \frac{1}{p!} \rmd x^{I_1} \cdots \rmd x^{I_p} \, \alpha_{I_p \cdots I_1}(x),$$
where the $\Z_2^{n+1}$-degree is $(p \mod 2, \deg(\alpha))$ and the $\Z_2^n$-degree is $\deg(\alpha) = \deg(\alpha_{I_p \cdots I_1}) + \deg(I_p) + \cdots \deg(I_1)$. Note that, in general, there are no top-forms on a $\Z_2^n$-manifold as $\rmd \zx$ is not nilpotent if $\zx$ is total degree odd.
\begin{warning}
By the degree of a differential $p$-form we will explicitly mean the $\Z_2^n$-degree.
\end{warning}
\begin{remark}
We can naturally think about the \emph{sheaf of differential forms} $|U| \mapsto \cO_{\mathbb{T}M}(|U|)$ for opens $|U| \subset |M|$. However, using the `reconstruction theorems' found in  \cite{Bruce:2018b}, we know that it is sufficient to work with the global algebraic picture when defining differential forms.
\end{remark}
The Cartan calculus, just as in the classical setting, consists of three vector fields on $\mathbb{T} M$, once a vector field on $M$ is specified. Using local coordinates $(x^I, \rmd x^J)$ we have the local expressions:\\
\begin{enumerate}
\itemsep1em
\item The \emph{de Rham derivative}: $\rmd := \rmd x^I \frac{\partial}{\partial x^I}$,  which has $\Z_2^{n+1}$-degree $(1, \mathbf{0})$,
\item The \emph{interior product}: $X = X^I(x) \frac{\partial}{\partial x^I} \rightsquigarrow i_X :=  X^I(x)\frac{\partial}{\partial \rmd x^I}$, which has $\Z_2^{n+1}$-degree $(1, \deg(X))$,
\item The \emph{Lie derivative}:  $X = X^I(x) \frac{\partial}{\partial x^I} \rightsquigarrow L_X := [\rmd, i_X] = X^I(x) \frac{\partial}{\partial x^I} + \rmd x^J \frac{\partial X^I(x)}{\partial x^J} \frac{\partial}{\partial \rmd x^I}$, which has $\Z_2^{n+1}$-degree $(0, \deg(X))$.
\end{enumerate}
Note that the $\Z_2^{n+1}$-degree of $\frac{\partial}{\partial x^I}$ and $\frac{\partial}{\partial \rmd x^J}$ are $(0, \deg(I))$ and $(1,\deg(J))$, respectively. The reader can easily check that these local expressions are well-defined, and so glue together to give global vector fields, using
\begin{align*}
& \frac{\partial}{\partial x^{I'}} = \left( \frac{\partial x^I}{\partial x^{I'}}\right)\frac{\partial}{\partial x^I} + \rmd x^J \left(\frac{\partial x^{J'}}{\partial x^J} \right) \frac{\partial^2 x^I}{\partial x^{J'}\partial x^{I'}}\frac{\partial}{\partial \rmd x^I},\\
& \frac{\partial}{\partial \rmd x^{I'}} = \left(\frac{\partial x^I}{\partial x^{I'}} \right) \frac{\partial}{\partial \rmd x^I}.
\end{align*}
From \cite[Proposition 2.3]{Covolo:2016}, which  states that vector fields (as $\Z_2^n$-graded derivations) are $\mathcal{J}$-adically continuous, we immediately have the following.
\begin{proposition}
Let $M$ be a $\Z_2^n$-manifold, and let us fix some vector field $X \in \Vect(M)$. Then the derivations
$$\rmd, \,i_X, \, L_X : \cO_{\mathbb{T}M }(|M|) \longrightarrow \cO_{\mathbb{T}M }(|M|)$$
are all $\mathcal{J}_{\mathbb{T}M}(|M|)$-adically continuous.
\end{proposition}
Using the local expressions it is easy to prove the following.
\begin{proposition}\label{prop:CartanCalc}
The de Rham differential, the interior product and the Lie derivative satisfy the following:
\begin{align*}
& 2 \,\rmd^2 =  [\rmd,\rmd] =0, && [i_X, i_Y] =0,\\
& [\rmd, L_X] =0, && [L_X,i_Y] = i_{[X,Y]},\\
& [L_X, L_Y] = L_{[X,Y]},
\end{align*}
for arbitrary vector fields $X,Y \in \Vect(M)$.
\end{proposition}
We will later need the following lemma. The proof follows from Proposition \ref{prop:CartanCalc} and the fact that $i_X f =0$ for any $X$ and $f \in\cO_M(|M|)$, thus we omit details.
\begin{lemma}\label{lem:dTwoForm}
Let $\alpha$ be a two-form on a $\Z_2^n$-manifold $M$. Then
\begin{align*} i_X i_Y i_Z \rmd \alpha &= X(i_Y i_Z \alpha) + (-1)^{\langle\deg(X), \deg(Y) + \deg(Z) \rangle}\, Y (i_Z i_X \alpha) + (-1)^{\langle \deg(Z), \deg(X)+ \deg(Y)\rangle }\, Z(i_X i_Y \alpha)-\\
&-i_Z i_{[X,Y]}\za-(-1)^{\langle \deg(X),\deg(Y)+\deg(Z)\ran} i_X i_{[Y,Z]}\za+(-1)^{\la\deg(Z),\deg(Y)\ran}
i_Y i_{[X,Z]}\za\,,
\end{align*}
for homogeneous vector fields $X,Y$ and $Z\in \Vect(M)$.
\end{lemma}
The de Rham complex is the cochain complex $\big(\Omega^*(M), \, \rmd \big)$ given by
$$0 \stackrel{\rmd}{\longrightarrow} \Omega^0(M)\stackrel{\rmd}{\longrightarrow} \Omega^1(M) \stackrel{\rmd}{\longrightarrow}  \Omega^2(M) \stackrel{\rmd}{\longrightarrow} \cdots\,,$$
which is not, in general, bounded from above.
\begin{remark}
One can also think in terms of a complex consisting of sheaves of differential forms. We will only consider the global complex here.
\end{remark}
The de Rham cohomology is defined as standard. Moreover, we have a version of the Poincar\'{e} lemma (for the proof see \cite[Theorem 5.10]{Covolo:2016}).
\begin{lemma}[Poincar\'{e}]\label{lem:Poincare}
Let $M$ be a $\Z_2^n$-manifold,  then the de Rham complex $(\Omega^*(M), \rmd)$ of $M$ is a resolution of the constant sheaf $\underline{\R}$.
\end{lemma}
The upshot of this lemma, just as in the standard case for smooth manifolds and supermanifolds, is that closed forms are locally exact. Moreover, it is known that the de Rham cohomology of $M$ reduces to the de Rham cohomology of the reduced manifold $|M|$ (see \cite[Remark 5.12]{Covolo:2016}).\par
In the same way as for supermanifolds (see \cite{Leites:1980}), there is a canonical map from the differential forms on $M$ to the differential forms on the reduced manifold $|M|$.
\begin{proposition}\label{prop:Kappa}
Let $M = (|M|, \cO_M)$ be a $\Z_2^n$-manifold, then there exists a canonical cochain map
$$\kappa : \big(\Omega^*(M), \, \rmd_M \big)\longrightarrow \big(\Omega^*(|M|), \, \rmd_{|M|}\big), $$
from the de Rham complex of $M$ to the de Rham complex of $|M|$.
\end{proposition}
\begin{proof}
We construct the map  $\kappa$ via local coordinates and then check that it is well-defined. We take a small enough $|U| \subset |M|$ so that we can employ local coordinates.  We will not write out all the required restrictions and define
\begin{align*}
\kappa(\rmd x^a) := \rmd x^a, && \kappa(\rmd \zx^A):= 0, && \kappa(f) := \epsilon_{|U|}(f),
\end{align*}
where $f \in \cO_M(|U|)$.  We now need to check that this definition is invariant under changes of coordinates.  Specifically, we need to see what happens to
\begin{align*}
& \rmd x^{a'} = \rmd x^b\frac{\partial x^{a'}}{\partial x^b} + \rmd \zx^B\frac{\partial x^{a'}}{\partial \zx^B}, && \rmd \zx^{A'} = \rmd x^b\frac{\partial \zx^{A'}}{\partial x^b} + \rmd \zx^B\frac{\partial \zx^{A'}}{\partial \zx^B}.
\end{align*}
One can easily observe that
\begin{align*}
&\kappa(\rmd x^{a'})= \rmd x^b \, \epsilon_{|U|}\left ( \frac{\partial x^{a'}}{\partial x^b}\right), && \kappa(\rmd \zx^{A'})= \rmd x^b\, \epsilon_{|U|}\left( \frac{\partial \zx^{A'}}{\partial x^b}\right) =0,
\end{align*}
remembering that each $\zx'$ is at least linear in (some) $\zx$ and so vanishes under $\epsilon$. The above gives exactly the expected transformation rules for $\rmd x$ and so the map $\kappa$ is globally well-defined. \par
We now proceed to examine what happens with the de Rham differentials. Let $\alpha \in \Omega^*(|U|)$ be an arbitrary differential from, and so, in general,  $\alpha =  \alpha(x, \zx, \rmd x, \rmd \zx)$. Directly,
$$\kappa\big ( \rmd_M  \alpha\big) = \kappa\left( \rmd x^a \frac{\partial \alpha}{\partial x^a} + \rmd \zx^A \frac{\partial \alpha}{\partial \zx^A}\right) = \rmd x^a\, \kappa\left(\frac{\partial \alpha}{\partial x^a} \right) = \rmd x^a\frac{\partial \kappa(\alpha)}{\partial x^a}. $$
Thus, we have shown that
$$\kappa \circ \rmd_M = \rmd_{|M|}\circ \kappa\,$$
as required.
\end{proof}
\begin{remark}\label{rem:Kappa}
It is clear that if $\alpha \in \Omega^*(M)$ is homogeneous with non-zero $\Z_2^n$-degree, then $\kappa(\alpha)$ =0.  Also, by construction, the homological degree of a differential $p$-form is preserved under $\kappa$, i.e., $\kappa\big ( \Omega^p(M)\big) \subset \Omega^p(|M|)$). These observations will be important when discussing (almost) symplectic structures.
 \end{remark}
To have a functor, we need to describe what $\mathbb{T}$ does to morphisms of $\Z_2^n$-manifolds $\phi :  M \rightarrow N$. This will allow us to define the pullback of differential forms. We do this via local coordinates. Let us employ local coordinates $x^I$ and $y^L$ on $M$ and $N$, respectively. We then write $\phi^*y^L := \phi^L(x)$. We then define
 $$(\mathbb{T}\phi)^* \rmd y^L := \rmd(\phi^L(x)) = \rmd x^I \frac{\partial \phi^L(x)}{\partial x^I}.$$
 Via standard arguments, we can glue this local description together to construct a morphism of $\Z_2^{n+1}$-manifolds.
 \begin{definition}
 Let $\phi : M \rightarrow N$ be a morphism of $\Z_2^n$-manifolds. The \emph{pullback of differential forms}  by $\phi$ is the $\Z_2^{n+1}$-graded $\Z_2^{n+1}$-commutative algebra morphism
 $$(\mathbb{T}\phi)^* : \Omega^*(N) \longrightarrow \Omega^*(M).$$
 We will follow classical notation write $\phi^* \alpha$ for any $\alpha \in \Omega^*(N)$.
 \end{definition}
 \begin{proposition}
 Let $\phi : M \rightarrow N$ be a morphism of $\Z_2^n$-manifolds. Then the pullback of differential forms by $\phi$ is a cochain map between the de Rham complexes
 $$\phi^* :  \big( \Omega^*(N), \rmd_N  \big) \longrightarrow \big( \Omega^*(M), \rmd_M  \big).$$
 \end{proposition}
 \begin{proof}
 It is sufficient to  prove this locally using coordinates. The result follows directly via the chain rule. Specifically
 $$\phi^* (\rmd_N \alpha) =  \phi^*(\rmd y^L)\, \phi^*\left(\frac{\partial \alpha}{\partial y^L} \right) =  \rmd x^I \frac{\partial \phi^L(x)}{\partial x^I}\, \phi^*\left(\frac{\partial \alpha}{\partial y^L} \right) =  \rmd x^I \frac{\partial \phi^*\alpha }{\partial x^I} =  \rmd_M(\phi^*\alpha),$$
for any $\alpha \in \Omega^*(N)$. Then $\phi^* \circ \rmd_N = \rmd_M \circ \phi^*$ as required.
 \end{proof}
\subsection{Almost symplectic and symplectic structures}
We are now in a position to make the main definition of this paper.
\begin{definition}
Let $M$ be a $\Z_2^n$-manifold. An \emph{almost symplectic structure} of $\Z_2^n$-degree $\deg(\omega) = \gamma \in \Z_2^n$ on $M$, is a  nondegenerate two-form $\omega$ of $\Z_2^n$-degree $\gamma$. An almost symplectic structure is said to be a \emph{symplectic structure} if it is closed, i.e., $\rmd \omega =0$. The pair $(M, \omega)$ is said to be an \emph{almost symplectic  $\Z_2^n$-manifold} if $\omega$ is an almost symplectic structure, and is said to be a \emph{symplectic  $\Z_2^n$-manifold} if $\omega$ is a symplectic structure.
\end{definition}
In local coordinates, any two-form is given by
$$\omega = \frac{1}{2!} \rmd x^I \rmd x^J \omega_{JI}(x),$$
where $\omega_{JI} =  -  (-1)^{\langle \deg(J), \deg(I) \rangle} \, \omega_{IJ}$. Under coordinates changes the components of a two-form transform as
\begin{equation}\label{egn:ChaCordTwoForm}
\omega_{J'I'}(x') = -  (-1)^{\langle \deg(I'), \deg(J) \rangle} \, \left( \frac{\partial x^J}{\partial x^{J'}}\right) \left( \frac{\partial x^I}{\partial x^{I'}}\right) \, \omega_{IJ}(x(x')).
\end{equation}
The nondegeneracy condition for an almost symplectic structure is, locally, the invertibility of the components of the two form, i.e.,  there exists a matrix $\omega^{IK}$ such that
$$\omega^{IK}\omega_{KJ} = \omega_{JK}\omega^{KI} = \delta_J^I.$$
If we are dealing with a symplectic structure, then the closure property can be written locally as
$$(-1)^{\langle \deg(I), \deg(K)\rangle}\, \frac{\partial \omega_{JI}}{\partial x^K} + (-1)^{\langle \deg(K), \deg(J) \rangle} \, \frac{\partial \omega_{IK}}{\partial x^J}+ (-1)^{\langle \deg(J), \deg(I) \rangle} \, \frac{\partial \omega_{KJ}}{\partial x^I} =0.$$
\begin{proposition}
Let $\omega_{KJ}$ be the local components of an almost symplectic structure. Then the inverse structure $\mathcal{P}^{IK} := (-1)^{\langle \deg(I),\deg(I) \rangle} \, \omega^{IK}$ has the symmetry
$$\mathcal{P}^{IK} = - (-1)^{\langle \deg(I), \deg(K) \rangle + \langle \deg(\omega), \deg(\omega)\rangle} \, \mathcal{P}^{KI}.$$
\end{proposition}
\begin{proof}
From the definition of invertibility and the symmetry of the components of a two-form we have
$$\omega^{IK}\omega_{KJ} = -  (-1)^{\langle \deg(I) + \deg(K) + \deg(\omega) ,\deg(J) + \deg(K) + \deg(\omega) \rangle + \langle \deg(K),\deg(J)\rangle + \Phi} \, \omega_{JK}\omega^{KI},$$
where $\Phi$ represents the symmetry of $\omega^{IK}$. When $I=J$ we need (up to the overall minus sign) 
$$ \Phi + \langle \deg(I) + \deg(K) + \deg(\omega) ,\deg(J) + \deg(K) + \deg(\omega) \rangle + \langle \deg(K),\deg(J)\rangle = 0\,.$$ 
Thus, 
$$\Phi = \langle \deg(I), \deg(I)\rangle + \langle \deg(K), \deg(K)\rangle  + \langle \deg(K), \deg(I)\rangle + \langle \deg(\omega), \deg(\omega)\rangle$$ 
and the proposition is established.
\end{proof}
If $\langle \gamma, \gamma\rangle = 0 / 1$, then $\omega$ is said to an \emph{even/odd almost symplectic structure}. Note that the inverse structures  of even and odd almost symplectic structures have very different symmetries under the exchange of the indices. Specifically, even structures are graded skewsymmetric, while odd structures are graded symmetric.
\begin{proposition}\label{prop:NonDegCon}
Let $(M,\omega)$ be an almost symplectic $\Z_2^n$-manifold. We set $p|\mathbf{q} = p|q_1, q_2, \cdots, q_N =: q_0,q_1, q_2, \cdots , q_N$. Then the non-degeneracy condition requires that
\renewcommand\labelenumi{(\roman{enumi})}
\begin{enumerate}
\item if $\deg(\omega) \neq \mathbf{0}$, then $q_i = q_j$ when $\gamma_i + \gamma_j = \deg(\omega)$,
\item if $\deg(\omega) = \mathbf{0}$ and $\langle \gamma_i , \gamma_i   \rangle =0$, then each $q_i$  must be an even integer, i.e., $q_i = 2 n_i$ for some integer $n_i$.
\end{enumerate}
\end{proposition}
\begin{proof}
As this is a local question, it is sufficient to examine the invertibility of the tensor $\omega_{IJ}$ in  some chosen, but arbitrary, set of local coordinates.  The $i,j$ block of $\omega$ is of $\Z_2^n$-degree $\gamma_i + \gamma_j + \deg(\omega)$. The entries of this matrix are in $\cO_M(|U|)$ for some ``small enough'' $|U| \subset |M|$. It is known that any such matrix is invertible if and only if the underlying real matrix $\epsilon_{|U|}(\omega)$ is itself invertible (see \cite{Covolo:2012} and for the super-case see \cite{Leites:1980}). Here $\epsilon : \cO_M  \rightarrow C^\infty_{|M|}$ is the canonical sheaf morphisms.  Under this projection, the only non-zero blocks are those for which $\gamma_i + \gamma_j = \deg(\omega)$.
\renewcommand\labelenumi{(\roman{enumi})}
\begin{enumerate}
\item Each block itself must be invertible and so a square matrix. Thus $q_i = q_j$ whenever $\gamma_i + \gamma_j = \deg(\omega)$.
\item The graded skewsymmetry of the almost symplectic structure implies that if $\langle \gamma_i, \gamma_i\rangle =0$, these  blocks are, in classical terminology,  skewsymmetric matrices. This directly implies that the $q_i$ must be even.
\end{enumerate}
\end{proof}
\begin{example}\label{exm:R2211}
Consider $\R^{2|2,1,1}$ equipped with global coordinates
$$(\underbrace{x, p}_{(0,0)}, ~\underbrace{z, w}_{(1,1)},~ \underbrace{\zx}_{(0,1)}, ~ \underbrace{\theta}_{(1,0)}).$$
Then $\omega_{(0,0)} = \rmd x\, \rmd p + \rmd z \, \rmd w +(\rmd \zx)^2 +(\rmd \theta)^2$ is a $\Z_2^2$-degree $\mathbf{0}$ symplectic structure. The similarity with even symplectic structures on supermanifolds is apparent.
\end{example}
\begin{example}\label{exm:R1111}
Consider $\R^{1|1,1,1}$ equipped with global coordinates
$$(\underbrace{x}_{(0,0)}, ~\underbrace{z}_{(1,1)},~ \underbrace{\zx}_{(0,1)}, ~ \underbrace{\theta}_{(1,0)}).$$
Then $\omega_{(1,1)} = \rmd x\, \rmd z + \rmd \theta \, \rmd \zx$, ~$\omega_{(0,1)} = \rmd x\, \rmd \zx + \rmd z \, \rmd \theta$, ~and  $\omega_{(1,0)} = \rmd x\, \rmd \theta + \rmd z \, \rmd \zx$   are  symplectic structures of $\Z_2^2$-degree $(1,1), (0,1)$ and $(1,0)$, respectively. The similarity with odd symplectic structures on supermanifolds is apparent.
\end{example}
\begin{proposition}\label{prop:RedStruc}
Let $(M, \omega_{\mathbf{0}})$ be an almost symplectic $\Z_2^n$-manifold with a $\Z_2^n$-degree zero almost symplectic structure. Then the reduced manifold $|M|$ is canonically almost symplectic. Moreover, if the two-form  $\omega_{\mathbf{0}}$ is a symplectic structure then the reduced manifold is a symplectic manifold.
\end{proposition}
\begin{proof}
We define $\omega_{\textrm{red}}:=  \kappa(\omega_{\mathbf{0}})$, which is a two-form on the reduced manifold $|M|$, a priori this could vanish  (see Remark \ref{rem:Kappa}). We need to argue that $\omega_{\textrm{red}}$ is non-degenerate.  In local coordinates, $\omega_{IJ}(x, \zx)$, as a matrix, is invertible if and only if each of its diagonal blocks are invertible (see the proof of Proposition \ref{prop:NonDegCon}. In particular, $\omega_{ab}(x, \zx)$ is non-zero and invertible,  and this is only the case if $\epsilon_{|U|}(\omega_{ab}(x, \zx))$ is invertible. The local invertibility is independent of the chosen coordinates, and hence, $\omega_{\textrm{red}}$ is an almost symplectic structure. \par.
If we now further assume that $\omega_{\mathbf{0}}$ is a symplectic structure, i.e., $\rmd_{M}\omega_{\mathbf{0}} =0$, then directly from Proposition \ref{prop:Kappa}, $\kappa(\rmd_M \omega_{\mathbf{0}}) = \rmd_{|M|}\kappa(\omega_{\mathbf{0}})=0$. Thus, we have a symplectic structure on $|M|$.
\end{proof}
\begin{example}
Continuing Example \ref{exm:R2211}, it is clear that $\kappa(\omega_{(00)}) =  \rmd x \, \rmd p$, which is the canonical symplectic structure on $\R^2$.
\end{example}
Just as in the classical setting, we have the natural notion of a symplectomorphism.
\begin{definition}
Let $(M_1, \omega_1)$ and $(M_2, \omega_2)$ be an almost symplectic $\Z_2^n$-manifolds with almost symplectic structures of the same $\Z_2^n$-degree.  The diffeomorphism $\phi :  M_1 \longrightarrow M_2$ is said to be a \emph{symplectomorphism} if $\phi^*\omega_2 = \omega_1$. The \emph{symplectomorphism group} $\textrm{Symp}(M, \omega)$ is the group of all symplectomorphisms $\phi : M \longrightarrow M$.
\end{definition}
\begin{remark}
As morphisms in the category of $\Z_2^n$-manifolds preserve the $\Z_2^n$-degree of functions, and similarly, for the associated pullback  of differential forms,  we cannot consider maps between (almost) symplectic $\Z_2^n$-manifolds with structures of different $\Z_2^n$-degree.
\end{remark}
We will consider infinitesimal symplectomorphisms and Hamiltonian vector fields in Subsection \ref{subsec:HamVect}.
\subsection{Hamiltonian vector fields and $\Z_2^n$-graded Poisson brackets}\label{subsec:HamVect}
In direct analogy with classic (almost) symplectic geometry, we have the notion of infinitesimal symmetries, or in other words vector fields whose Lie derivative annihilates the (almost) symplectic form, i.e., symplectic vector fields.
\begin{definition}
Let $(M, \omega)$ be an almost symplectic $\Z_2^n$-manifold. A vector field $X \in \Vect(M)$ is said to be a \emph{symplectic vector field} if $L_X\omega =0$.  We will denote the set of  symplectic vector fields by $\Vect_{\omega}(M)$.
\end{definition}
\begin{proposition}
Let $(M, \omega)$ be an almost symplectic $\Z_2^n$-manifold. The set of  symplectic vector fields $\Vect_{\omega}(M)$ forms a $\Z_2^n$-Lie algebra under the Lie bracket of vector fields on $M$.
\end{proposition}
\begin{proof}
First, we need to show that the set of symplectic vector fields has the structure of a $\R$-vector space. This is clear from the linear nature of the Lie derivative, i.e., $L_{X + cY} = L_X + c\, L_Y$ for all vector fields $X$ and $Y$, and $c \in \R$. Thus, any linear combination of  symplectic vector fields is a symplectic vector field.  Similarly, $L_{[X,Y]} = [L_X, L_Y]$ for all vector fields $X$ and $Y$, implies that if $X$ and $Y$ are  symplectic vector fields, then is also  $[X,Y]$ a symplectic vector field.
\end{proof}
\begin{remark}
Note that symplectic vector fields may have non-zero $\Z_2^n$-degree.   The Lie algebra of $\Z_2^n$-degree zero symplectic vector fields is interpreted as the Lie algebra associated with the symplectomorphism  group $\textrm{Symp}(M, \omega)$. We will not make use of this here and so refrain from presenting details.
\end{remark}
\begin{definition}
Let $(M, \omega)$ be an almost symplectic $\Z_2^n$-manifold. A vector field $X \in \Vect(M)$ is said to be a \emph{Hamiltonian vector field} if  it is a symplectic vector field, i.e., $L_X\omega =0$, and $i_X\omega \in \Omega^1(M)$ is exact.
\end{definition}
\begin{remark}
If $\omega$ is closed, i.e., it is a symplectic structure, then $i_X\omega$  being exact implies that $X$ is a symplectic vector field.  For an almost symplectic structure, we need both conditions for a vector field to be a Hamiltonian vector field.
\end{remark}
We will adopt the standard notation $i_{X_f}\omega = \rmd f$, where $f \in \cO_M(|M|)$. Balancing the $\Z_2^n$-degree we see that $\deg(X_f) =  \deg(\omega) + \deg(f)$.  \par
\begin{definition}\label{def:PoissonBracket}
Let $(M, \omega)$ be an almost symplectic $\Z_2^n$-manifold, where $\deg(\omega) =  \gamma \in \Z_2^n$. Then the associated \emph{almost Poisson bracket} is the bilinear mapping
$$\{- ,-\}_\omega : \cO_M(|M|)\times  \cO_M(|M|) \longrightarrow \cO_M(|M|),$$
given by
$$\{f,g\}_\omega := i_{X_f} i_{X_g}\omega = X_f(g).$$
If $\omega$ is a symplectic structure then we speak of the associated \emph{Poisson bracket}.
\end{definition}
In local coordinates, the (almost) Poisson bracket is given by
$$\{f,g \}_\omega = (-1)^{\langle \deg(f) + \gamma,  \gamma\rangle  + \langle \deg(f) , \deg(I)\rangle} \, \mathcal{P}^{IJ}(x)\frac{\partial f}{\partial x^J} \frac{\partial g}{\partial x^I},$$
where $\mathcal{P}^{IJ} = (-1)^{\langle \deg(I), \deg(I)}\, \omega^{IJ}$. Specifically,
\begin{equation}\label{eqn:PoisX}
\{x^I, x^J \}_\omega = - (-1)^{\langle \deg(I), \gamma\rangle} \, \mathcal{P}^{IJ}(x).
\end{equation}
The almost Poisson bracket has the expected properties.
\begin{proposition}\label{prop:AlmPoissbrk}
Let $(M, \omega)$ be an almost symplectic $\Z_2^n$-manifold. Then the associated almost Poisson bracket
\begin{enumerate}
\itemsep1em
\item is of $\Z_2^n$-degree $\gamma$, i.e., $\deg(\{f,g \}_\omega) =  \deg(f) + \deg(g) + \gamma$;
\item is shifted skewsymmetric, i.e., $\{ f,g\}_\omega = - (-1)^{\langle \deg(f) + \gamma, \deg(g) + \gamma \rangle}\, \{g,f  \}_\omega$, and
\item satisfies the Leibniz rule, i.e.,  $\{ f,gh\}_\omega =  \{f,g\}_\omega \, h + (-1)^{\langle\deg(f) + \gamma, \deg(g) \rangle }\, g \, \{ f,h \}_\omega,$
\end{enumerate}
for all $f,g$ and $h \in \cO_M(|M|)$.
\end{proposition}
\begin{proof}\
\begin{enumerate}
\itemsep1em
\item This is clear as the $\Z_2^n$-degrees of the objects are $\deg(i_{X_f}) = \deg(f)+ \gamma$, and similar for $i_{X_g}$, and $\deg(\omega) = \gamma$. Thus, $\deg(i_{X_f} i_{X_g}\omega) =  \deg(f) + \deg(g) + \gamma$.
\item Directly,  $\{f,g \}_\omega = i_{X_f} i_{X_g}\omega =  - (-1)^{\langle \deg(f) + \gamma, \deg(g) + \gamma \rangle}\, i_{X_g}i_{X_f} \omega = - (-1)^{\langle \deg(f) + \gamma, \deg(g) + \gamma \rangle}\, \{g,f  \}_\omega$.
\item This follows from a short calculation:
\begin{align*}
\{f,gh\}_\omega & = i_{X_f}i_{X_gh}\omega = i_{X_f}\big(\rmd g \, h + g \, \rmd h \big)\\
& = X_f(g) \,h +(-1)^{\langle \deg(f) + \gamma , \deg(g) + \gamma\rangle } \, g\, X_f(h)\\
& = \{ f,g\}_\omega \,h +(-1)^{\langle \deg(f) + \gamma , \deg(g) + \gamma\rangle } \, g\, \{f,h \}_\omega.
\end{align*}
\end{enumerate}
\end{proof}
\begin{proposition}
Let $(M, \omega)$ be an almost symplectic $\Z_2^n$-manifold (with $\deg(\omega)= \gamma$). Then
$$[X_f, X_g] =  X_{\{ f,g \}_\omega},$$
for all $f$ and $g \in \cO_M(|M|)$.
 \end{proposition}
 \begin{proof}
 Recall that for Hamiltonian vector fields $L_{X_f}\omega =0$ and $i_{X_f}\omega = \rmd f$. Directly
 \begin{align*}
 i_{[X_f, X_g]}\omega  &= [L_{X_f}, i_{X_g}]\omega =  L_{X_f}(i_{X_g}\omega) \\                      &= L_{X_f}(\rmd g) = \rmd (X_f g) =  \rmd (\{ f,g\}_ \omega) \\
 &= i_{X_{\{f,g \}_\omega}}.
 \end{align*}
 \end{proof}
Just as with the classical and supergeometric cases, if the almost symplectic structure is closed, so a symplectic structure, we have a graded version of the Jacobi identity.
\begin{theorem}\label{thm:PoissbrkJac}
Let $(M, \omega)$ be a symplectic $\Z_2^n$-manifold, then the associated Poisson bracket further satisfies the $\Z_2^n$-graded Jacobi identity
$$\{ f, \{ g,h\}_\omega \}_\omega =  \{\{f,g \}_\omega ,h  \}_\omega + (-1)^{\langle \deg(f) + \gamma, \deg(g)+ \gamma \rangle} \, \{ g, \{f,h \}_\omega \}_\omega,$$
for any homogeneous $f,g$ and $h\in \cO_M(|M|)$.
\end{theorem}
\begin{proof}
Using Lemma \ref{lem:dTwoForm} together with Definition \ref{def:PoissonBracket} (and the symmetry of the almost Poisson bracket) we observe that
\begin{align*}
i_{X_f} i_{X_g}i_{X_h} \rmd \omega &=  \{ f, \{ g,h\}_\omega\}_\omega + (-1)^{\langle \deg(f) + \gamma , \deg(g) + \deg(h)\rangle}\, \{g, \{h,f\}_\omega \}_\omega  \\
&+ (-1)^{\langle \deg(h) + \gamma , \deg(f) + \deg(g)\rangle}\, \{h, \{f,g\}_\omega \}_\omega  \\
&= \{ f, \{ g,h\}_\omega \}_\omega -  \{\{f,g \}_\omega ,h  \}_\omega - (-1)^{\langle \deg(f) + \gamma, \deg(g)+ \gamma \rangle} \, \{ g, \{f,h \}_\omega \}_\omega,
\end{align*}
for any homogeneous $f,g$ and $h\in \cO_M(|M|)$. As we have a symplectic structure,  $\rmd \omega =0$ and we observe that the Jacobi identity holds. 
\end{proof}
\begin{remark}
For a symplectic $\Z_2^n$-manifold, the pair $(\cO_M(|M|),\{-,- \}_\omega)$ should be referred to as a $\gamma$-shifted $\Z_2^n$-Poisson algebra.  As far as we know, the closely related notion of a colour Poisson (super)algebra was first proposed by Trostel motivated by applications in generalised statistics in quantum theory (see \cite{Trostel:1984}).
\end{remark}
\begin{example}
Continuing Example \ref{exm:R1111}, the Poisson brackets associated with the given symplectic structures are
\begin{align*}
& \{f ,g\}_{(1,1)} = (-1)^{\langle (1,1), \deg(f)\rangle} \, \frac{\partial f}{\partial z} \frac{\partial g}{\partial x} -  \frac{\partial f}{\partial x}\frac{\partial g}{\partial z} - (-1)^{\langle (1,0), \deg(f)\rangle} \,\frac{\partial f}{\partial \theta} \frac{\partial g}{\partial \zx}  - (-1)^{\langle (0,1), \deg(f)\rangle} \,\frac{\partial f}{\partial \zx} \frac{\partial g}{\partial \theta},\\[1em]
& \{f ,g\}_{(0,1)} = - (-1)^{\langle (0,1), \deg(f)\rangle} \, \frac{\partial f}{\partial \zx} \frac{\partial g}{\partial x} -  \frac{\partial f}{\partial x}\frac{\partial g}{\partial \zx} + (-1)^{\langle (1,1), \deg(f)\rangle} \,\frac{\partial f}{\partial z} \frac{\partial g}{\partial \theta}  - (-1)^{\langle (1,0), \deg(f)\rangle} \,\frac{\partial f}{\partial \theta} \frac{\partial g}{\partial z},\\[1em]
&\{f ,g\}_{(1,0)} = - (-1)^{\langle (1,0), \deg(f)\rangle} \, \frac{\partial f}{\partial \theta} \frac{\partial g}{\partial x} -  \frac{\partial f}{\partial x}\frac{\partial g}{\partial \theta} + (-1)^{\langle (1,1), \deg(f)\rangle} \,\frac{\partial f}{\partial z} \frac{\partial g}{\partial \zx}  - (-1)^{\langle (0,1), \deg(f)\rangle} \,\frac{\partial f}{\partial \zx} \frac{\partial g}{\partial z}.
\end{align*}
 The reader should note the similarity of above Poisson brackets  with the antibracket (odd Poisson or Schouten bracket) as found in the BV-BRST formalism of gauge theory. 
\end{example}
\subsection{Canonical symplectic structures on cotangent bundles}
Just as in the classical case, the cotangent bundle of a $\Z_2^n$-manifold comes with a canonical $\Z_2^n$-degree zero symplectic form. Moving to supermanifolds, the parity shifted cotangent bundle comes with a canonical odd symplectic structure.  A very similar situation occurs in  $\Z_2^n$-geometry where we have a multitude of different shifts in the grading.\par
We will describe the cotangent bundle via local coordinates. We equip $\sT^*M$ with coordinates $(x^I, p_J)$, where the  $\Z_2^n$-degrees of the coordinates are $\deg(x^I) = \deg(I)$ and $\deg(p_J) = \deg(J)$ and the admissible changes of coordinates are of the form
\begin{align*}
& x^{I'} = x^{I'}(x), && p_{J'} = \left( \frac{\partial x^I}{\partial x^{J'}}\right) \, p_I.
\end{align*}
We then have inherited coordinates on $\mathbb{T}\sT^*M$
$$\big( \underbrace{x^I}_{(0, \deg(I))}, ~  \underbrace{p_J}_{(0, \deg(J))}, ~ \underbrace{\rmd x^K}_{(1, \deg(K))}, ~ \underbrace{\rmd p_L}_{(1, \deg(L))} \big),$$
and the admissible changes of coordinates are
\begin{align*}
& \rmd x^{K'} =  \rmd x^K \left( \frac{\partial x^{K'}}{\partial x^K}\right),\\
& \rmd p_{J'} = \left( \frac{\partial x^{J}}{\partial x^{J'}}\right) \rmd p_J + \rmd x^K\left( \frac{\partial x^{K'}}{\partial x^K}\right)\left( \frac{\partial^2 x^I}{\partial x^{K'} \partial x^{J'} }\right)p_I.
\end{align*}
\begin{proposition}\label{prop:ConSymZero}
Let $M$ be a $\Z_2^n$-manifolds. Then the cotangent bundle $\sT^* M$
comes equipped with a canonical $\Z_2^n$-degree zero symplectic structure $\omega_{\mathbf{0}}$.
\end{proposition}
\begin{proof}
In some chosen coordinate system we claim that the symplectic structure is given by $\omega_{\mathbf{0}} = \rmd x^I \rmd p_I.$ The $\Z_2^n$-degree is clear, as is the nondegeneracy and the fact that it is closed. We just need to check that the two-form is well-defined, i.e., it does not depend on choice of coordinates.  Directly,
$$ \omega'_{\mathbf{0}} =  \rmd x^{I'} \rmd p_{I'}
=  \rmd x^I \rmd p_I + \rmd x^{I'} \rmd x^{K'} \left( \frac{\partial^2 x^L}{\partial x^{K'} \partial x^{I'}} \right)p_L. $$
Note that the final term consists of a contraction over something skewsymmetric and symmetric in indices and so it vanishes. Thus, $\omega'_{\mathbf{0}} = \omega_{\mathbf{0}}$ and so the proposition is established.
\end{proof}
For any $\gamma \in \Z_2^n$ (we will include $\mathbf{0}$ to be consistent later) we have a $\gamma$-degree shifted cotangent bundle $\Pi_\gamma \sT^*M$, which we will define via coordinates. Essentially, the $\Z_2^n$-degree of the momentum is shifted by $\gamma$. To do this, we treat $\Pi_\gamma$ as a formal object of $\Z_2^n$-degree $\gamma$ in local expressions.  In particular
$$\Pi_\gamma p_{J'} =  \Pi_\gamma \left( \left(\frac{\partial x^I}{\partial x^{J'}}  \right)p_I \right) = (-1)^{\langle \deg(I) + \deg(J'), \gamma \rangle}\, \left(\frac{\partial x^I}{\partial x^{J'}}  \right) \Pi_\gamma p_I.$$
Then defining, $p^\gamma_J := \Pi_\gamma p_J$, we define local coordinates on $\Pi_\gamma \sT^*M$
$$( \underbrace{x^I}_{\deg(I)}, ~ \underbrace{p^\gamma_J}_{\deg(J) + \gamma }  ),$$
and the admissible changes of coordinates are
\begin{align*}
& x^{I'} = x^{I'}(x), && p^\gamma_{J'} = (-1)^{\langle \deg(I) + \deg(J'), \gamma \rangle}\, \left(\frac{\partial x^I}{\partial x^{J'}}  \right)  p^\gamma_I.
\end{align*}
Then $\mathbb{T}\sT^*M$ comes with induced coordinates
$$\big( \underbrace{x^I}_{(0, \deg(I))}, ~ \underbrace{p^\gamma_J}_{(0, \deg(J) + \gamma)}, ~ \underbrace{\rmd x^K}_{(1, \deg(K))} ,~ \underbrace{\rmd p^\gamma_L}_{(1, \deg(L) + \gamma)} \big),$$
where the admissible changes of coordinates are as before, except the shifted momenta now transform as
$$\rmd p^\gamma_{J'} = (-1)^{\langle \deg(I) + \deg(J'), \gamma \rangle}\,\left( \frac{\partial x^{J}}{\partial x^{J'}}\right) \rmd p^\gamma_J +(-1)^{\langle \deg(I) + \deg(J'), \gamma \rangle}\, \rmd x^K\left( \frac{\partial x^{K'}}{\partial x^K}\right)\left( \frac{\partial^2 x^I}{\partial x^{K'} \partial x^{J'} }\right)p^\gamma_I. $$
\begin{theorem}\label{thm:CoTanStruct}
Let $M$ be a $\Z_2^n$-manifold. Then the $\gamma$-shifted cotangent bundle $\Pi_\gamma \sT^*M$ comes equipped with a $\Z_2^n$-degree $\gamma$ canonical symplectic structure $\omega_\gamma$, which is  locally given by $\omega_\gamma = (-1)^{\langle \deg(I), \gamma \rangle} \, \rmd x^I \rmd p^\gamma_I$.
\end{theorem}
\begin{proof}
We define $\omega_\gamma := \Pi_\gamma \omega_{\mathbf{0}} =  \Pi_\gamma (\rmd x^I \rmd p^\gamma_I) = (-1)^{\langle \deg(I), \gamma \rangle} \, \rmd x^I \rmd p^\gamma_I$. The $\Z_2^n$-degree is clear as is the nondegeneracy and the fact that the two-form is closed. The proof that the two-form is well-defined under coordinate transformations is almost identical to the proof of Proposition \ref{prop:ConSymZero} and so we omit details.
\end{proof}
In adapted coordinates $(x^I, p_J^\gamma)$ the canonical Poisson brackets are
\begin{equation} \label{eqn:CoTanBra}
\{f,g \}_\omega = (-1)^{\langle \deg(f)+ \gamma , \deg(I) + \gamma\rangle  + \langle \deg(I), \deg(I) \rangle} \,\frac{\partial f}{\partial p_I^\gamma}\frac{\partial g}{\partial x^I}  -  (-1)^{\langle\deg(f), \deg(I) \rangle}\, \frac{\partial f}{\partial x^I}\frac{\partial g}{\partial p_I^\gamma},
\end{equation}
which should be compared with the classical Poisson brackets as found in mechanics and the antibracket  as found in the BV-BRST formalism.
\subsection{$\Z_2^n$-graded gauge and Hamiltonian systems}
The notion of gauge and Hamiltonian systems naturally generalises to the current setting. As compared to supermanifolds, we have a lot more choice in how we assign the degrees. In particular, we have more freedom in how we choose functions as being of total degree even/odd.
\begin{definition}
Let $(M, \omega)$ be an almost symplectic $\Z_2^n$-manifold with an almost symplectic structure of $\Z_2^n$-degree $\gamma$. A \emph{homological potential} is a section $\Theta \in \cO_M(|M|)$ of total degree odd/even if $\gamma$ is  even/odd, and $\{\Theta, \Theta \}_\omega = 0$. A triple $(M, \omega, \Theta)$ is referred to as a \emph{(non-degenerate) gauge system}.
\end{definition}
Note that the condition $\{\Theta, \Theta\}_\omega =0$ is non-trivial. The nomenclature here is borrowed from physics and in particular the BV-BRST formalism of gauge theory (see \cite{Lyakhovich:2004}).
\begin{proposition}
Let $(M, \omega, \Theta)$ be a gauge system with $\omega$ being a symplectic structure. Then $M$ comes equipped with a canonical homological vector field $Q_\Theta \in  \Vect(M)$, i.e., $2(Q_\Theta)^2 = [Q_\Theta, Q_\Theta] =0$.
\end{proposition}
\begin{proof}
We define $Q_\Theta :=  X_\Theta = \{ \Theta, -\}_\omega$, which is of $\Z_2^n$-degree $\deg(\Theta) + \gamma$, noting that, by definition, it is of odd total degree. We need to check that this vector field `squares to zero'. This follows from the Jacobi identity for the Poisson bracket (and so we need a symplectic rather than an almost symplectic structure).  Let $f \in \cO_M(|M|)$ be arbitrary, then
$$(Q_\Theta)^2 f =  \{\Theta, \{\Theta , f\}_\omega \}_\omega = \{ \{\Theta, \Theta \}_\omega, f\}_\omega -  \{\Theta, \{ \Theta, f\}_\omega \}_\omega\,. $$
Thus, $(Q_\Theta)^2 f = \frac{1}{2} \{ \{ \Theta, \Theta\}_\omega,f \}_\omega =0$.
\end{proof}
\begin{definition}
Let $(M,\omega, \Theta)$ be a gauge system with a symplectic structure $\omega$. The associated \emph{standard cochain complex} is $(\cO_M(|M|), Q_\Theta)$.
\end{definition}
In complete analogy with a supercomplex, we have what we shall call a $\Z_2^n$-complex
\begin{equation*}
\leavevmode
\begin{xy}(0,60)*+{\cO_M(|M|)_{\gamma_i}}="a"; (40,60)*+{\cO_M(|M|)_{\gamma_i + \deg(\Theta)+ \gamma}}="b";%
{\ar@<1.ex>@/^1.pc/|{Q_\Theta}"a";"b"};%
{\ar@<1.ex>@/^1.pc/|{Q_\Theta} "b";"a"};
\end{xy}
\end{equation*}
for all $i = 0,1, \cdots , N$, and here $\deg(Q_\Theta) = \deg(\Theta)+ \gamma$. A homogeneous function $f  \in \cO_M(|M|)_{\gamma_i}$ is said to be \emph{$Q_\Theta$-closed} if $Q_\Theta f =0$, and \emph{$Q_\Theta$-exact} if there exists a $g \in\cO_M(|M|)_{\gamma_i + \deg(\Theta)+\gamma}$ such that $f = Q_\Theta g$. The homogeneous kernel and image of  $Q_\Theta$ are defined in the obvious way.  We can then define the \emph{$i$-th standard cohomology group} as
$$\mathrm{H}^i_{st}(Q_\Theta) := \Ker(Q_\Theta)_{\gamma_i}\setminus \textrm{Im}(Q_\Theta)_{\gamma_i} \,,$$
i.e., the space of $Q_\Theta$-closed but not $Q_\Theta$-exact functions on $M$ of degree $\gamma_i$. The algebra structure of $\cO_M(|M|)$ induced a bilinear map
$$\mathrm{H}^i_{st}(Q_\Theta) \times \mathrm{H}^j_{st}(Q_\Theta) \mapsto \mathrm{H}^{i+ j}_{st}(Q_\Theta),$$
where $i+j$ is counted mod $N$. \par
The notion of a Hamiltonian system is similarly defined.
\begin{definition}
Let $(M, \omega)$ be an almost symplectic $\Z_2^n$-manifold with an almost symplectic structure of $\Z_2^n$-degree $\gamma$. A \emph{Hamiltonian} is a section $H \in \cO_M(|M|)$ of total degree even/odd if $\gamma$ is  even/odd. A triple $(M, \omega, H)$ is referred to as a \emph{(non-degenerate) Hamiltonian system}.
\end{definition}
Note that $\{H,H \}_\omega =0$ automatically.
\begin{remark} Thinking of dynamical systems in this context (being intentionally loose):
\begin{enumerate}
\item For a gauge system the ``gauge parameter'' $\lambda$, defined via $f \mapsto f + \lambda \, \{ \Theta, f\}_\omega$ is of $\Z_2^n$-degree $\deg(\lambda) = \deg(\Theta) + \gamma$, which is odd.
\item Similarly, for a Hamiltonian system the ``time'' $t$,    defined via $f \mapsto f + t \, \{ H, f\}_\omega + O(t^2)$ is of $\Z_2^n$-degree $\deg(t) = \deg(H) + \gamma$, which is even, but not necessarily $\mathbf{0}$.
\end{enumerate}
For even/odd symplectic (more generally Poisson/Schouten) supermanifolds we only have degree $0$ time and degree $1$ gauge parameters. For the more general $\Z_2^n$-case, we have a lot more freedom with the $\Z_2^n$-degree of the evolution parameters.
\end{remark}
\subsection{BV-like Laplacians}
It is well-known that the antibracket in the BV-BRST formalism is generated by an odd Laplacian via the failure of the Leibniz rule for second-order operators.  A similar phenomenon occurs in the setting of $ \Z_2^n$-geometry.  We restrict attention to $\Pi_\gamma \sT^*M$ with $\gamma$ odd, i.e., $\langle \gamma, \gamma \rangle =1 $. Then by Theorem \ref{thm:CoTanStruct}, we know that we have an odd Poisson bracket (see \eqref{eqn:CoTanBra} for the local expression).  In natural coordinates, we define the \emph{BV-like Laplacian} as
\begin{equation}\label{eqn:BVLap}
\Delta^\gamma|_{|U|} := \frac{\partial^2}{\partial x^I \partial p_I^\gamma} + \textnormal{first-order terms}.
\end{equation}
The exact nature of the first-order terms is irrelevant for the following. However, they are required if $\Delta^\gamma$ is to be well-defined. For example, one may use an affine connection to compensate for the first-order terms generated by coordinate transformations of the second-order term. Alternatively, we may consider $M =  \R^{p|\mathbf{q}}$ and linear changes of coordinates and simply drop the first-order terms.
\begin{proposition}
Let $\Pi_\gamma \sT^*M$  be such that $\gamma$ odd. Then there exists a BV-like Laplacian  of the form \eqref{eqn:BVLap} that generates the Poisson bracket as the failure or anomaly to the graded Leibniz rule, i.e.,
$$\Delta^\gamma(fg) = \Delta^\gamma(f) \, g  + (-1)^{\langle \gamma, \deg(f) \rangle} \, f \, \Delta^\gamma(g) + (-1)^{\langle \deg(f) + \gamma , \gamma\rangle} \, \{f,g \}_\omega,$$
where $f$ and $g$ are global sections on $\Pi_\gamma \sT^*M$.
\end{proposition}
\begin{proof}
As we are interested in the Leibniz rule we can safely ignore the first-order terms of the BV-like Laplacian.  Then directly in local coordinates
\begin{align*}
\Delta^\gamma(fg) &= \frac{\partial ^2 f}{\partial x^I \partial p_I^\gamma}\, g +  (-1)^{\langle \gamma, \deg(f)\rangle }  f\, \frac{\partial ^2 g}{\partial x^I \partial p_I^\gamma}\\
&+ (-1)^{\langle \deg(I), \deg(f) + \deg(I) + \gamma \rangle} \, \frac{\partial f}{\partial p_I^\gamma} \frac{\partial g}{\partial x^I} + (-1)^{\langle \deg(f),  \deg(I) + \gamma \rangle} \,\frac{\partial f}{\partial x^I} \frac{\partial g}{\partial p_I^\gamma}.
\end{align*}
Comparing the final two terms with \eqref{eqn:CoTanBra}, we see that we have the Poisson bracket, up to the overall factor of $(-1)^{\langle  \deg(f)+ \gamma, \gamma\rangle}$, provided $\langle \gamma, \gamma \rangle =1$, as we suppose.
\end{proof}
\begin{remark}
In order to fully generalise the BV-formalism to the setting of $\Z_2^n$-geometry one needs an understanding of the theory of integration on $\Z_2^n$-manifolds.  The theory of integration in this setting is, at the time of writing, in its infancy. Currently, only very low dimensional examples are understood. For progress in this direction the reader can consult \cite{Poncin:2016} and the appendix of \cite{Bruce:2020aa}.
\end{remark}
\subsection{The $\Z_2^n$-graded Darboux theorem}
Let $M = (|M|, \cO_M)$ be a $\Z_2^n$-manifold and let $\zw$ be a symplectic structure of $\Z_2^n$-degree $\zg$ on $M$. An open set $|U|\subset |M|$ will be called a \emph{Darboux coordinate neighbourhood} for $(M,\zw)$ if $|U|$ is a coordinate neighbourhood which admits homogeneous coordinates $(q^1,\dots, q^{j},p_1,\dots,p_j,y_1,\dots y_k)$ (not necessarily written in canonical order) such that
\be\label{Darboux}
\zw |_{|U|}=\sum_{i=1}^j  (-1)^{\langle \gamma, \deg(i)\rangle}\, \rmd q^i\, \rmd p_i  +\sum_{l=1}^k\frac{\ze_l}{2}(\rmd y_l)^2\,,
\ee
where $\ze_l=\pm 1$. This should be compared with the Darboux theorem in symplectic supergeometry, see for example Schwarz \cite{Schwarz:1996}, proofs of which are sketched by Kostant \cite{Kostant:1977} and Shander \cite{Shander:1983}. We remark that Khudaverdian gave a `simple' proof of the Darboux theorem for odd symplectic supermanifolds in \cite{Khudaverdian:2004}. Rothstein in \cite{Rothstein:1991} gives an alternative version of the Darboux theorem for even symplectic supermanifolds based on Batchelor's theorem.\par
We will need some preliminaries before we state and prove the relevant Darboux theorem. The \emph{tangent space} of $M$ at $m \in |M|$, denoted $\sT_{m}M$, is the $\Z_2^n$-graded $\R$-vector space of $\Z_2^n$-graded $\R$-linear derivations $\cO_{M,\, m} \rightarrow \R$.  Recall that any $\Z_2^n$-graded derivation $X : \cO_M(|U|) \rightarrow \cO_M(|U|)$ induces a $\Z_2^n$-graded derivation at the level of stalks $X|_{m} : \cO_{M,\, m} \rightarrow \cO_{M,\, m}$, of the same $\Z_2^n$-degree as $X$ (assuming homogeneity), for any $m \in |U|$. We denote $\epsilon_{m} : \cO_{M,\, m} \rightarrow C^\infty_{m}$ as the algebra morphism induced by pullback to the reduced manifold of $M$. Furthermore the evaluation morphism at $m$ we denote as $\textrm{ev}_{m} : C^\infty_{m} \rightarrow \R$. We then define $u \in \sT_{m}M$ as the tangent vector
$$u := (\textrm{ev}_{m} \circ \epsilon_{m} \circ X|_{m}). $$
\begin{definition}
Let $M$ be a $\Z_2^n$-manifold and $|U| \subset |M|$ be open. A section $X \in \mathcal{T}M(|U|)$ is said to be \emph{non-degenerate at} $m \in |U|$  if the associated tangent vector $u\in \sT_p M$ is non-zero, i.e., not the zero vector.
\end{definition}
The \emph{cotangent space} of $M$ at $m \in |M|$, denoted $\sT^*_m M$, is the $\Z_2^n$-graded $\R$-vector space defined as
$$\sT_m^*M := \InHom_{\, \R}\big( \sT_m M, \R \big),$$
where $\InHom_{\, \R}$ is the internal homs in the category of $\Z_2^n$-graded $\R$-vector spaces. Every one-form $\alpha$ on $|U| \subset |M|$ gives rise to a covector in $\sT_m^*M$ for any $m \in |U|$. First, we define, for any $X$ vector field on $|U|$, $\alpha(X) := (-1)^{\langle \deg(\alpha), \deg(X)\rangle}\, i_X \alpha$. Thus we consider a one-form as a $\R$-linear map
$$\alpha(-) : \Der\big( \cO_M(|U|) \big) \longrightarrow \cO_M(|U|).$$
We can then work at the level of germs at $m \in |U|$
$$\alpha|_m(-) : \Der\big( \cO_{M,m} \big) \longrightarrow \cO_{M,m}.$$
A tangent vector $u \in \sT_m M$ can be considered as a constant derivation using the coordinate basis, i.e., $u = u^I \partial_I|_m$ where each $u^I \in \R$. This can then be considered as a constant element of $\Vect(U)$ for ``small enough'' $|U| \subset |M|$. We then define
$$\alpha_m(u) :=  \mathrm{ev}_m \circ \epsilon_m (\alpha|_m(u)) = (-1)^{\langle  \deg(\alpha), \deg(X) \rangle}\, (i_u\alpha)|_m\,,$$
and thus we consider $\alpha_m \in \sT_m^*M$.
\begin{definition}
Let $M$ be a $\Z_2^n$-manifold and let $|U| \subset |M|$ be an open neighbourhood of $m \in |M|$. A one-form $\alpha \in \Omega^1(U)$ is said to be \emph{non-degenerate} at $m \in |M|$ if the associated covector $\alpha_m$ is non-zero, i.e., not the zero covector.
\end{definition}
\begin{definition} Let $M = (|M|, \cO_M)$ be a $\Z_2^n$-manifold and let $|U| \subset |M|$ be open.
\begin{enumerate}
\item A (finite) collection of vector fields $\{ X_1, X_2, \cdots, X_s \}$ on $|U|$ is said to be \emph{linearly independent} at $m \in |U|$ if the associated set of tangent vectors at $m$ is linearly independent.
\item A (finite) collection of one-forms $\{ \alpha_1, \alpha_2, \cdots, \alpha_s \}$ on $|U|$ is said to be \emph{linearly independent} at $m \in |U|$ if the associated set of covectors at $m$ is linearly independent.
\end{enumerate}
\end{definition}
\begin{lemma}\label{lem:IndpOneForms}
Let $\alpha$ and $\beta$ be one-forms on $|U|$ that are of different $\Z_2^n$-degrees and non-degenerate at $m \in |U|$. Then $\alpha$ and $\beta$ are linearly independent at $m$.
\end{lemma}
\begin{proof}
Via assumption $\alpha_m$ and $\beta_m$ are non-zero and belong to different graded sectors of $\sT_m^*M$ and so they are linearly independent covectors. Thus, $\alpha$ and $\beta$ are linearly independent at $m$.
\end{proof}
\begin{lemma}\label{lem:SkewMap}
An almost symplectic structure $\omega \in \Omega^2(M)$ (of $\Z_2^n$-degree $\gamma$) gives rise to a non-degenerate $\R$-bilinear  $\Z_2^n$-graded skewsymmetric map (of the same $\Z_2^n$-degree as $\omega$)
$$\omega_m : \sT_m M \times \sT_m M \rightarrow \R,$$
for any $m \in |M|$.
\end{lemma}
\begin{proof}
As before, we consider tangent vectors $u,v \in \sT_m M$ as giving rise to constant vector fields on $|U|$ with $m \in |U|$.  Then we define
$$\omega_m(u,v) := (-1)^{\langle \gamma, \deg(u) + \deg(v)\rangle}\, \textrm{ev}_m\left(\epsilon_m \big( i_u i_v\omega|_m \big)  \right).$$
The properties are evident as they follow from the properties of the almost symplectic structure.
\end{proof}
\begin{lemma}\label{lem:LinDepZero}
Let $\omega$ be a $\Z_2^n$-degree $\mathbf{0}$ almost symplectic structure on a $\Z_2^n$-manifold $M$. Furthermore, let $X$ and $Y$ be vector fields on $|U|$ that are even and of the same $\Z_2^n$-degree and non-degenerate at $m \in |U|$.  Then $X$ and $Y$ are linearly dependant if and only $\omega_m(X_m, Y_m)=0$.
\end{lemma}
\begin{proof}
If $X$ and $Y$ are linearly dependant at $m \in |U|$, then $Y_m = c \, X_m$ for some $c \in \R$ (non-zero). This implies that
$$\omega_m(X_m, Y_m) = c \, \omega_m(X_m, X_m) = - c \, \omega_m(X_m, X_m)=0,$$
where we have used the (graded) skewsymmetry. \par
 In the other direction, we observe that $\big(\sT_m M\big)_{\zg_l}$, with $\gamma_l\in \Z_2^n$ being an even degree, is a symplectic $\R$-vector space. Indeed, the non-degeneracy of the almost symplectic form implies that its components in local coordinates are ``block invertible''. In turn, this means that at any point we can consider the even sectors of the tangent space as a symplectic vector space.  Thus, we can always work with the symplectic bilinear form in a Darboux basis. As a simple illustration, assume that  $\big(\sT_m M\big)_{\gamma_l}\simeq \R^2$. Then in this case we can always find a basis such that $\omega_m(X_m, Y_m) = - X_m^1 Y_m^2 + X_m^2 Y_m^1$. Now assuming that $\omega_m(X_m, Y_m)=0$  implies that $X_m$ and $Y_m$ are proportional (remembering we assume $X$ and $Y$ are non-degenerate at $m$) and so linearly dependant. Thus, $X$ and $Y$ are linearly dependant at $m$. 
\end{proof}
\begin{lemma}\label{lem:OnlyOddCord}
Let  $\omega$ be a homogeneous symplectic structure of even total degree on $M = (\star, \Lambda_{\textrm{odd}})$ where $\Lambda_{\textrm{odd}}$ is a $\Z_2^n$-Grassmann algebra with just odd generators. Then we can find generators of $\Lambda_{\textrm{odd}}$  $(q^1, \cdots ,q^j , p_i, \cdots p_j, y_1, \cdots y_k)$, where $\deg(p_i) = \deg(q^i)+ \gamma$, such that
$$\omega = \sum_{i=1}^j  (-1)^{\langle \gamma, \deg(i)\rangle}\, \rmd q^i\, \rmd p_i  +\sum_{l=1}^k\frac{\ze_l}{2}(\rmd y_l)^2 + J^2 \Omega^2(M)\,$$
where $\ze_l = \pm 1$ and the generators $y$ are non-zero only if $\gamma = \mathbf{0}$. Here $J$ is the ideal generated by the generators.
\end{lemma}
\begin{proof}
Let us employ arbitrary odd generators $\zx^A$ (of the required number and degrees).  Observe that
$$\omega = \rmd \zx^A \, \rmd \zx^B f_{BA} + J^2 \Omega^2(M)\, $$
with each  $f_{BA} \in \R$ not all zero as $\omega$ is nondegenerate. Note that as $\omega$ is even that there can be no term in $J\Omega^2(M)$. Furthermore, we require $\deg(A) + \deg(B) = \gamma$ and $f_{BA} = - (-1)^{\langle \deg(A), \deg(B)\rangle }\, f_{AB}$. The question then becomes one of finding the appropriate form of the real matrix $f_{BA}$ via a linear coordinate transformations.  We are not concerned with the exact form of the term in $J^2 \Omega^2(M)$ at this juncture. Thus, we consider $f_{BA}$ as a non-degenerate homogeneous bilinear form on the $\Z_2^n$-graded $\R$-vector space  $\sT^\star M :=V = \bigoplus_{ \gamma_p \in(\Z_2^n)_{\textrm{odd}}} V_{\gamma_p}$ (see Lemma \ref{lem:SkewMap}). We can thus use a modified Gram--Schmidt process to bring the bilinear form into the desired form.\par
We will denote the nondegenerate pairing as $(-,-): V \times V \rightarrow V$. Let $v \in V$ be a non-zero  homogeneous vector, then there are two possibilities
\begin{enumerate}
\renewcommand{\theenumi}{\roman{enumi}}%
\item $(v,v) = 0$, or
\item $(v,v) \neq 0$.
\end{enumerate}
\noindent \underline{Case (i)}~Let $v\in V$ be homogeneous and we set $v= v_1$, and look for a $w_1 \in V$ such that $(v_1, w_1) = 1$.  This can always be found as the pairing is nondegenerate and $v_1$ is non-zero.  Note that this implies $\deg(w_1) = \deg(v_1) + \gamma$ (which is odd). Clearly $v_1$ and $w_1$ are linearly independent as $(v_1,v_1)=0$. Thus, $(v_1, w_1) = (-1)^{\langle \deg(v_1), \gamma \rangle}\, (w_1, v_1)$. We then pass to the complimentary space  $\langle v_1, w_1 \rangle^{\perp} := \{u \in V ~|~(a v_1 + bw_1 , u)=0,  ~~ a,b \in \R \}$ and repeat the process until all $v$ have been exhausted. Note that since $v_1$ and $w_1$ are homogeneous, the complimentary space $\langle v_1, w_1 \rangle^{\perp}$ is spanned by homogeneous vectors. Because at each step we pass to complimentary spaces we end up with a set of linearly independent and homogeneous vectors. Let us also observe that if $\gamma \neq \mathbf{0}$ then the process stop at this first stage as case (ii) is impossible.

\medskip

\noindent \underline{Case (ii)}~We now find a homogeneous $v = z_1$ such that $(z_1, z_1)= \pm 1$. This can always be found by rescaling a selected vector. We then pass to the complimentary space $\langle z_1 \rangle^{\perp} := \{ y \in V, ~ | ~ (z_1, y) =0\}$, and repeat the  process until all $v$ have been exhausted. Because at each step we pass to complimentary spaces we end up with a set of linearly independent and homogeneous vectors.
\medskip

In conclusion, the nondegenerate pairing on $V$ can then be cast into the canonical form
$$ \begin{pmatrix}
0 & \delta_i^j & 0\\
(-1)^{\langle \deg(i), \gamma\rangle} \, \delta_i^j & 0 & 0\\
0 & 0 & \epsilon_m\delta^m_l\\
\end{pmatrix}.
$$
\end{proof}
\begin{theorem}\label{thm:Darboux}
 Any  symplectic $\Z_2^n$-manifold $(M,\zw)$, with symplectic structure of $\Z_2^n$-degree $\zg$, can be covered by Darboux coordinate neighbourhoods (see \eqref{Darboux}). The coordinates $y_l$ can appear only in the case $ \zg= \mathbf{0}$.
\end{theorem}
\begin{warning} Our conventions with differential forms are closes to  Deligne's conventions with differential forms on a supermanifold. In the proceeding proof, when stating that a one-form is even or odd we will be referring to the total $\Z_2^n$-degree and not the $\Z_2^{n+1}$-degree. Specifically, if $q$ is an even/odd coordinate, then $\rmd q$ is said to be even/odd.  On the other hand, for the skew-products of one-forms we use the $\Z_2^{n+1}$-degree.
\end{warning}
\begin{proof}\

\smallskip
\textbf{Step 1} Let $p_1$ be a coordinate, $p_1\in\cO_M(|U|)$, for a neighbourhood  $|U|\subset|M|$ of $m\in|M|$ so that $\deg(p_1)+\zg$ is even. If there is no such coordinate, go directly to Step 2.\par
Let $X$ be a vector field in this neighbourhood such that $i_X\zw= (-1)^{\langle \gamma, \deg(1) \rangle}\, \rmd p_1$ (we will not write out all the required restrictions), which is non-degenerate at $m$. The vector field is by assumption even, so according to \cite[Proposition 5.6]{Covolo:2016c}, there is an even coordinate $q^1$ (in a neighbourhood of $m$, perhaps a bit smaller, $|V|\subset|U|$) such that $X=\pa_{q^1}$.\par
 Next let $Y$ be a vector field on $|U|$ such that $i_Y\zw= (-1)^{\langle \deg(1), \deg(1)\rangle}\,\rmd q^1$. Clearly, $Y$ is non-degenerate at $m$ and of degree $\deg(p_1)$. Note that $\deg(q^1)=\deg(X)=\zg+\deg(p_1)$, so $\rmd q^1$ and $\rmd p_1$ are of different $\Z_2^n$-degrees, and hence linearly independent at $m$, if $\zg\ne \mathbf{0}$ (see Lemma \ref{lem:IndpOneForms}). \par
  If $\deg(q^1)=\deg(p_1)$, then $\zg=\mathbf{0}$, and $\zw$, $\rmd q^1$ and $\rmd p_1$ are all even. In this case, however, $X$ and $Y$, thus $\rmd p_1$ and $\rmd q^1$, are also linearly independent at $m$. This follows from Lemma \ref{lem:LinDepZero}. In particular,
  $$\omega_m(X_m, Y_m) =  \textrm{ev}_m\circ \epsilon_m\big( i_Xi_Y \omega|_m\big)= \textrm{ev}_m\circ \epsilon_m \big(i_X \rmd q^1 |_m  \big) =1 \neq 0,$$
 and thus we have linear independence of $X$ and $Y$ at $m$, and since $\omega_m$ is non-degenerate we conclude that $\rmd p_1$ and $\rmd q^1$ are also linearly independent at $m$.\par
In any case, $ (-1)^{\langle \gamma, \deg(1) \rangle}\, \rmd q^1 \, \rmd p_1$ is a two-form of degree $\zg$, that is non-degenerate at $m$. Moreover, $L_X\zw=L_Y\zw=0$, as $X$ and $Y$ are (locally) Hamiltonian vector fields,  and $[X,Y]=0$ as
$$i_{[X,Y]}\zw=(L_Xi_Y\pm i_YL_X)\zw=L_Xi_Y\zw=\pm \rmd\big( i_{\pa_{q^1}}\,\rmd q^1\big)  =0\,.$$
In conclusion, $X$ and $Y$ are ($\Z_2^n$-graded) commuting vector fields spanning a two-dimensional distribution. Then, via  \cite[Theorem 5.7]{Covolo:2016c}, we have a neighbourhood $|U_1|\subset|V|$  of $m$ and coordinates
$(q^1,p_1,z^1,\dots,z^r)$ such that $X=\pa_{q^1}$ and $Y=\pa_{p_1}$. We then set
$$\zw_1=\zw - (-1)^{\langle \gamma, \deg(1)\rangle}\, \rmd q^1 \, \rmd p_1\,.$$
Clearly, $\rmd \zw_1=0$. We then observe that $\zw_1$ cannot contain a term associated with $\rmd q^1$ and $\rmd p_1$, as $i_X\zw_1=i_Y\zw_1=0$. Moreover, we know that $\zw_1$ does not depend on $q^1,p_1$, since $L_X\zw_1=L_Y\zw_1=0$; this can easily be checked using the local expression for the Lie derivative.\par
In consequence, we can view $\zw_1$ as a symplectic form in coordinates $(z^1,\dots,z^r)$ of degree $\zg$. Then we repeat this process  iteratively as long as we are able to find  a coordinate $p_i$  such that $\zg+\deg(p_i)$ is even.
We end up with coordinates $(z^1,\dots,z^r)$ such that $\zg+\deg(z^s)$ are all odd and $\zw$ is of the form
$$\zw=\sum_{i=1}^j \, (-1)^{\langle \gamma, \deg(i)\rangle}\, \rmd q^i \, \rmd p_i+\zw_j\,,$$
where $\zw_j$ is a symplectic form of degree $\zg$ in coordinates $(z^1,\dots,z^r)$.

\medskip

\textbf{Step 2.} \
Assume now that on a $\Z^n_2$-manifold $M'$ we have a symplectic form $\zw$ of degree $\zg$ and $\zg+\deg(z^s)\in\Z^n_2$ are odd for all homogeneous coordinates $(z^1,\dots,z^r)$. If $\zg$ is odd, then all coordinates are even; a contradiction, because a two-form in even coordinates must be even. Thus, for symplectic structures of odd total degree, the procedure terminates at Step 1. and the corresponding Darboux theorem is established. If $\zg$ is even, then all coordinates are odd and the problem becomes completely algebraic.  In particular, we consider the $\Z_2^n$-manifold $M'= (\star , \Lambda_{\textrm{odd}})$, where  $\Lambda_{\textrm{odd}}$ is the $\Z_2^n$-Grassmann algebra with the required number of odd generators only. This means that the space $\zW^2(M')$ of two-forms on $M'$ is finite-dimensional as we only have nilpotent coordinates. \par
We choose homogeneous coordinates $(q^1,\dots, q^{j},p_1,\dots,p_j, ,y_1,\dots y_k)$ (all odd and not necessarily in canonical order) such that Lemma \ref{lem:OnlyOddCord} holds, and we define
\be\label{Darboux2}
\zw_0= \sum_{i=1}^j (-1)^{\langle \gamma, \deg(i) \rangle} \rmd q^i \, \rmd p_i+\sum_{l=1}^k\frac{\ze_l}{2}(\rmd y_l)^2\,,
\ee
where $\ze_l=\pm 1$.  Note that $\deg(q^i)=\zg + \deg(p_i)$, and we have coordinates $y_l$ only in the case $\zg=\mathbf{0}$. We then define $\zw_1=\zw-\zw_0$, which is clearly a closed two-form. As $\zw_1$ is even, it follows that $\zw_1\in J^2\zW^2(M')$, where $J$ is the ideal generated by formal variables. Hence, there are elements $f^1,\dots,f^j,g_1\dots,g_j,h_1,\dots h_k$ of $\Lambda_{\textrm{odd}}$ belonging to $J^3$,
$$\deg(f^i)=\deg(p_i)\,,\quad \deg(g_i)=\deg(q^i)\,,\quad \deg(h_l)=\deg(y_l)\,,$$
such that (via the Poincar\'{e} lemma (see Lemma \ref{lem:Poincare}))
\be\label{Darboux4}
\zw=\zw_0+\xd\left(\sum_{i=1}^j\left( f^i\, \rmd  p_i+  \rmd q^i g_i \right)+\sum_{l=1}^k h_l \, \rmd y_l\right)\,.
\ee
The coordinates $h_l$ can appear only in the case $\zg=\mathbf{0}$. Let us define now the new set of homogeneous coordinates:
$$\tilde{p}_i=p_i+g_i\,, \quad \tilde{q}^i=q^i+f^i\,, \quad \tilde{y}_l=y_l+h_l\,.$$
It is easy to see that
\be
\tilde\zw_0=\sum_{i=1}^j  (-1)^{\langle\gamma, \deg(i)\rangle}\rmd \tilde q^i \, \rmd \tilde p_i+\sum_{l=1}^k\frac{\ze_l}{2}(\rmd \tilde y_l)^2\,.
\ee
differs from $\zw$ by a close two-form belonging to $J^4\zW^2(M')$. As the space $\zW^2(M')$ is finite-dimensional, the corresponding recurrent procedure ends up at $0$ and the Darboux theorem is proven.
\end{proof}

\section{Discussion and possible applications}
\subsection{Recap} In this paper, we have shown that the foundational aspects of symplectic geometry generalise to the setting of $\Z_2^n$-manifolds. In particular, we have an associated (shifted) $\Z_2^n$-graded Poisson bracket that has all the expected properties, shifted cotangent bundles come with canonical symplectic structures, and we have a version of the Darboux theorem. As compared with supersymplectic geometry, there is a lot more freedom with the degree of symplectic structures other than just even or odd. However, for certain aspects, symplectic structures on $\Z_2^n$-manifolds can be distinguished by their total degree as  even or odd. For example, the local components of a symplectic structure are (graded) skewsymmetric in its indices, while the inverse structures are skewsymmetric or symmetric depending on the symplectic structure being even or odd.  Furthermore, we have a BV-like Laplacian than generates the Poisson bracket if and only if the symplectic structure is odd.
\subsection{Mackenzie theory}
By Mackenzie theory we mean the rich tapestry of ideas related  to double and multiple vector bundles, double Lie algebroids, Lie bialgebroids and their natural links with Poisson geometry.   Recall that double vector bundles (in the category of smooth manifolds) have  a natural $\Z_2^2$-superization (see \cite{Covolo:2016a}). In particular, $\Pi_{\Z_2^2}\sT^* E$, where $E \rightarrow M$ is a vector bundle, can be equipped with homogeneous coordinates of the form
$$(\underbrace{x^a}_{(0,0)}, \, \underbrace{\zx^i}_{(1,0)},\, \underbrace{p_b}_{(1,1)}, \, \underbrace{\pi_j}_{(0,1)})\,,$$
and comes with a canonical (non-degenerate) Poisson bracket of $\Z_2^2$-degree $(1,1)$ locally give by
$$\{F,G \} =  (-1)^{\langle (1,1), \deg(F)\rangle}\, \frac{\partial F}{\partial p_a}\frac{\partial G}{\partial x^a} - \frac{\partial F}{\partial x^a}\frac{\partial G}{\partial p_a} - (-1)^{\langle (1,0), \deg(F)\rangle}\, \frac{\partial F}{\partial \pi_i} \frac{\partial G}{\partial \zx^i} -  (-1^{\langle (0,1), \deg(F)\rangle})\, \frac{\partial F}{\partial \zx^i}\frac{\partial G}{\partial \pi_i}\,.$$
The pair of (linear) homological potentials
$$ \Theta_{(0,1)} := -  \zx^i Q_i^a(x)p_a + \half \zx^i \zx^jQ_{ji}^k(x)\pi_k\, , \quad \textnormal{and} \quad \Theta_{(1,0)} := - Q^{ia}(x)p_a \pi_i + \half \zx^k Q_k^{ij}\pi_j \pi_i\,,$$
can be shown to be equivalent to the structure of a Lie algebroid on $E$ and $E^*$, respectively. Furthermore, the compatibility condition $\{ \Theta_{(1,0)}, \Theta_{(0,1)}\} =0$ leads to a pair of commuting homological vector fields $Q_{(1,0)}$ and $Q_{(0,1)}$. These two homological vector fields encode a Lie bialgebroid structure on the pair $(E, E^*)$ as first defined by Mackenzie and Xu (see \cite{Mackenzie:1994}). This is completely analogous to the results of Roytenberg (see \cite{Roytenberg:2002}) and Voronov (see \cite{Voronov:2012}). In the super-setting we have a pair of anticommuting homological vector fields. The potential mismatch of signs in the compatibility conditions is compensated for by the different commuting/anticommuting nature of the coordinates. Rather than present details here, we will defer a careful study of Mackenzie's double Lie algebroids and Lie bialgebroids etc. in the setting of $\Z_2^n$-geometry to a future publication.
\subsection{Parastatistics}
Being speculative, degree zero symplectic structures on $\Z_2^n$-manifolds may find application in classical mechanics that include parabosons and parafermions  in terms of Green components (\cite{Green:1953,Volkov:1959}). Indeed, $\Z_2^n$-gradings have long been recognised as important in quantum mechanical parastatistics (see for example Tolstoy \cite{Tolstoy:2014} and references therein). Earlier works in the direction of understanding the classical notion parastatistics include \cite{Kalany:1970,Mostafazadeh:1996a,Mostafazadeh:1996b}. As an illustration, consider a system consisting of a  single non-relativistic paraboson and a single non-relativistic parafermion of order two, which we write in Green components as
$$Q = q^1 + q^2, \quad \Psi = \psi^1 + \psi^2\, ,$$
which are subject to the commutation rules
\begin{align*}
&q^i q^i = + q^i q^i, && q^i q^j = - q^j q^i & \textnormal{if}~ i \neq j\, ,\\
&\psi^{\hat i} \psi^{\hat i} = - \psi^{\hat i} \psi^{\hat i}, &&\psi^{\hat i} \psi^{\hat j} = + \psi^{\hat j} \psi^{\hat i} & \textnormal{if}~ \hat i \neq \hat j\, , \\
&q^{i} \psi^{\hat i} = + \psi^{\hat i} q^{ i}, && q^{ i} \psi^{\hat j} = - \psi^{\hat j} q^{ i} & \textnormal{if}~  i \neq \hat j\, .
\end{align*}
 We have some freedom in how we chose the relative statistics, and we here choose \emph{normal relative statistics} (cf. \cite{Greenberg:1965}). This generalises the standard commutation rules between a boson and a fermion -- parabosons and parafermions are `relative parabosons'. These commutation rules can be implemented by assigning a $\Z_2^4$-degree
 \begin{align*}
 & \deg(q^i)=\deg(i) = \left\{ \begin{matrix} (0,1,0,1)\quad \textnormal{if}~i=1\\(1,0,0,1)\quad \textnormal{if}~i=2  \end{matrix}\right.\\
 & \deg(\psi^{\hat i})=\deg(\hat i) = \left\{ \begin{matrix} (0,1,1,1)\quad \textnormal{if}~\hat i=1\\(1,0,1,1)\quad \textnormal{if}~\hat i=2  \end{matrix}\right.
 \end{align*}
We then consider these degrees of freedom as coordinates on the purely formal $\Z_2^4$-manifold $M = (\star, \R[[q, \psi]])$. We then consider the phase space to be the cotangent bundle $\sT^*M =  (\star , \R[[q,\psi, p, \pi]])$, where $\deg(p_i) = \deg(i)$ and  $\deg(\pi_{\hat i}) = \deg(\hat i)$. In general, the physical phase space will only be a $\Z_{2}^{n}$-submanifold of  the cotangent bundle of the configuration $\Z_{2}^{n}$-manifold. As our intention is not to discuss general mechanics in the presence of various constraints, we will not elaborate further. \par
The degree $\mathbf{0}$ Poisson bracket then has the form
$$\{ f,g\} = (-1)^{\langle \deg(f) , \deg(i)\rangle }\, \left( \frac{\partial f}{\partial p_i}\frac{\partial g}{\partial q^i} -  \frac{\partial f}{\partial q^i}\frac{\partial g}{\partial p_i}\right) + (-1)^{\langle \deg(f) , \deg(\hat i)\rangle }\, \left( \frac{\partial f}{\partial \pi_{\hat i}}\frac{\partial g}{\partial \psi^{\hat i}} +  \frac{\partial f}{\partial \psi^{\hat i}}\frac{\partial g}{\partial \pi_{\hat i}}\right).$$
The only physically relevant degrees of freedom are $Q = q^1 + q^2$,  $\Psi = \psi^1 + \psi^2$, $P = p_1 + p_2$ and $\Pi = \pi_1 + \pi_2$. Note that we cannot say if these pairwise commute or anticommute as they are inhomogeneous in $\Z_2^4$-degree. Any potentially physically relevant Hamiltonian must be a function of these variables and be $\Z_2^4$-degree $\mathbf{0}$. We propose the following simple Hamiltonian as an explicit example,
\begin{align*}
H &=  \frac{1}{2m} P^2 + \frac{k}{2}Q^2 + \frac{\lambda}{2}\big( \Psi \Pi - \Pi \Psi\big)\\
&= \frac{1}{2m} \delta^{ij}p_j p_i + \frac{k}{2}q^i q^j\delta_{ji} + \lambda \, \psi^{\hat i} \pi_{\hat i}\,,
\end{align*}
where $m, k$ and $\lambda \in \R$ have the standard interpretation of masses and a coupling constant. This Hamiltonian is the obvious generalisation of the Hamiltonian of the harmonic (super)oscillator.  Being slack with what we exactly mean by a phase trajectory\footnote{It is clear that one needs `external parameters' that carry non-trivial $\Z_2^4$-graded parameters. To do this rigorously one should think about a time parametrised functor of points. We will not dwell on this here and direct the reader to \cite{Bruce:2017} for a discussion of supermechanics.}, the phase dynamics are given by Hamilton's equations
\begin{align*}
& \frac{\rmd q^i}{\rmd t}(t) = \delta^{ij}\left(\frac{p_j(t)}{m}\right)\,, && \frac{\rmd p_j}{\rmd t}(t) = - k\, q^i(t)\delta_{ij}\,,\\
& \frac{\rmd \psi^{\hat i}}{\rmd t}(t) = - \lambda\, \psi^{\hat i}(t)\,,&& \frac{\rmd \pi_{\hat j}}{\rmd t}(t) = \lambda \, \pi_{\hat j}(t)\,.
\end{align*}
Bringing this back into the physical degrees of freedom we have
\begin{align*}
& \frac{\rmd Q}{\rmd t}(t) = \left(\frac{P(t)}{m}\right)\,, && \frac{\rmd P}{\rmd t}(t) = - k\, Q(t)\,,\\
& \frac{\rmd \Psi}{\rmd t}(t) = - \lambda\, \Psi(t)\,,&& \frac{\rmd \Pi}{\rmd t}(t) = \lambda \, \Pi(t)\,.
\end{align*}
Note that we cannot, as expected, directly apply Dirac's quantisation procedure due to the presence of constraints as indicated  by the phase dynamics of the parafermion. None-the-less, this example suggests that symplectic $\Z_2^n$-manifolds are useful in the theory of parastatistics. Moreover, one should note that \emph{any} sign rule for a finite number of objects can be encoded in a $\Z_2^n$-grading, with the standard scalar product determining the sign rule, for a sufficiently large $n$ (\cite[Theorem 2.1]{Covolo:2016}). Thus,  any finite number of paraparticles with exotic commutation rules between their Green components can be accommodated in symplectic $\Z_2^n$-geometry.

\section*{Acknowledgements}
 The authors thank the anonymous referees for their invaluable comments and suggestions. A.J.~Bruce thanks Norbert Poncin for many helpful discussions about $\Z_2^n$-geometry and related subjects.   J.~Grabowski acknowledges that his  research was funded by the  Polish National Science Centre grant HARMONIA under the contract number 2016/22/M/ST1/00542.

\end{document}